%% file: 1809_bmms.tex
\RequirePackage{etoolbox}
\csdef{input@path}{%
	{img/}
}%

\documentclass[12pt]{article}
\usepackage{amsmath}
\usepackage{graphicx,psfrag,epsf}
\usepackage{enumerate}
\usepackage{natbib}
\usepackage{url} 
\usepackage[skip=5pt]{caption}
\usepackage[normalem]{ulem}
\usepackage{amsthm}
\usepackage{amsmath}
\usepackage{natbib}
\usepackage[colorlinks,citecolor=blue,urlcolor=blue,filecolor=blue,backref=page]{hyperref}
\usepackage{graphicx}
\usepackage{algorithm2e}
\usepackage{amssymb}
\usepackage{amsbsy}
\usepackage{float}
\usepackage{comment}
\usepackage{mathtools}
\usepackage{color,soul}
\usepackage{blkarray}
\usepackage{xcolor}
\usepackage{tikz}

\RequirePackage{doi}

\newcommand{\wip}{\noindent\hl{[\#]} }
\newcommand{\bmms}{\textit{BM\&Ms} }
\newcommand{\iidsim}{\overset{iid}{\sim}} 

\newcommand{\textalt}[1]{\quad\text{#1}\quad}
\newcommand{\sigmasq}{\sigma^2}

\usepackage{bbm} 
\newcommand{\bml}{b}

\newcommand{\bm}[1]{\mathbf{#1}}

\usepackage{scalerel,stackengine}
\stackMath
\newcommand\reallywidehat[1]{%
	\savestack{\tmpbox}{\stretchto{%
			\scaleto{%
				\scalerel*[\widthof{\ensuremath{#1}}]{\kern-.6pt\bigwedge\kern-.6pt}%
				{\rule[-\textheight/2]{1ex}{\textheight}}
			}{\textheight}%
		}{0.5ex}}%
	\stackon[1pt]{#1}{\tmpbox}%
}

\usepackage{amsthm}

\theoremstyle{definition}
\newtheorem{definition}{Definition}[section]
\newtheorem{proposition}{Proposition}[section]
\newtheorem{corollary}{Corollary}[section]

\newcommand{\blind}{0}

\begin{document}

\def\spacingset#1{\renewcommand{\baselinestretch}%
{#1}\small\normalsize} \spacingset{1}


\if0\blind
{
  \title{\bf Bayesian Modular and Multiscale Regression}
  \author{Michele Peruzzi\thanks{
    This research was partially supported by grants from the United States Office of Naval Research and National Institutes of Health.
}\hspace{.2cm}\\
    Department of Decision Sciences, Bocconi University\\
    Department of Statistical Science, Duke University\\
    $ $\vspace{-0.4cm}\\
    David B. Dunson \\
    Department of Statistical Science, Duke University}
  \maketitle
} \fi

\if1\blind
{
  \bigskip
  \bigskip
  \bigskip
  \begin{center}
    {\Large \bf Bayesian Modular and Multiscale Regression}
\end{center}
  \medskip
} \fi

\bigskip
\begin{abstract}
We tackle the problem of multiscale regression for predictors that are spatially or temporally indexed, or with a pre-specified multiscale structure, with a Bayesian modular approach. The regression function at the finest scale is expressed as an additive expansion of coarse to fine step functions. Our Modular and Multiscale (M\&M) methodology provides multiscale decomposition of high-dimensional data arising from very fine measurements. Unlike more complex methods for functional predictors, our approach provides easy interpretation of the results. Additionally, it provides a quantification of uncertainty on the data resolution, solving a common problem researchers encounter with simple models on down-sampled data. We show that our modular and multiscale posterior has an empirical Bayes interpretation, with a simple limiting distribution in large samples. An efficient sampling algorithm is developed for posterior computation, and the methods are illustrated through simulation studies and an application to brain image classification. Source code is available as an \texttt{R} package at \url{github.com/mkln/bmms}.
\end{abstract}

\noindent%
\vfill

\newpage
\spacingset{1.45} 

\input{01_introduction_brain}

\input{02_modular_general_180805}

\input{03_modular_linear_180809}

\input{06_computation_180810}

\input{07_data_analysis}

\input{08_discussion}

\bibliographystyle{hapalike}
\bibliography{biblio}

\appendix

\input{91_appendix_180810}

\pagebreak

\Huge Supplementary material
\normalsize

\input{92_bmms_supplement}

\end{document}

%% file: 01_introduction_brain.tex
\section{Introduction} \label{bmms:intro}
Modern researchers routinely collect very high-dimensional data that are spatially and/or temporally indexed, with the intention of using them as inputs in regression-type problems. A simple example with a binary output is \textit{time-series classification}; analogously, brain images can be used as diagnostic tools. In these cases, prediction of the outcome variable and interpretation of the results are the main goals, but obtainining clear interpretation and accurate prediction simultaneously is notoriously difficult. In fact, adjacent measurement locations tend to be highly correlated, with possibly huge numbers of measurements at a very high resolution. 

In these settings, directly inputing such data into usual regression methods leads to poor results. In fact, methods for dimensionality reduction that do not take advantage of predictor structure can have poor performance in estimating regression coefficients and sparsity patterns when the dimension is huge and predictors are highly correlated (see e.g. Figure \ref{brain}). Theoretical guarantees in such settings typically rely on strong assumptions on sparsity, low linear dependence, and high signal-to-noise (\citealt{irrepresentable, wasserman}; \citealt[Chapter~7]{Buhlmann11}; \citealt{ieee17}).

\begin{figure}
	\centering
	\includegraphics[width=0.7\textwidth]{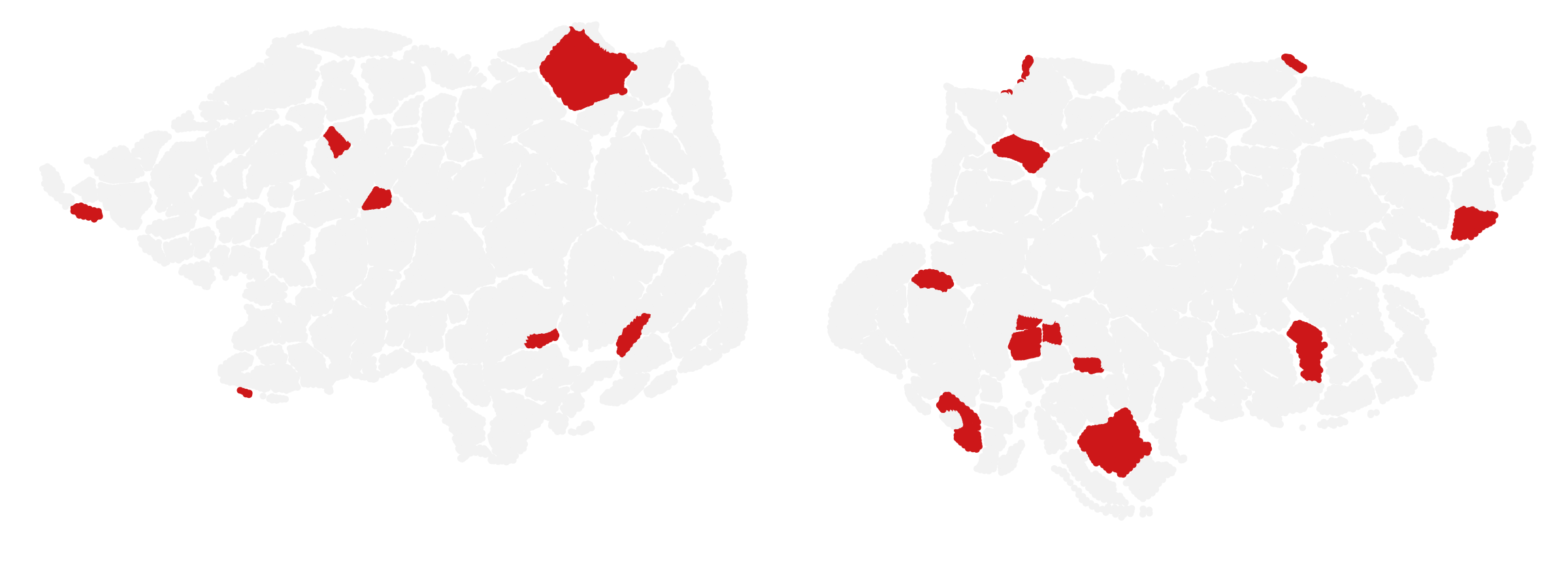}
	\caption{Brain regions parcellated according to \cite{gordon333} and selected in a Lasso model for a gender classification task.}\label{brain} 
\end{figure}

The above problems can be alleviated by down-sampling the data to lower resolutions before analysis. This is an appealing option because of the potential for huge dimensionality reduction. However, any specific resolution choice might be perceived as ad-hoc, and hide patterns at different scales that could instead be highlighted by a multiresolution approach. This problem cannot be solved by methods that somewhat take into account predictor structure, but only act on a single measurement scale, like the group Lasso of \cite{glasso}, the fused Lasso of \cite{fused}, the Bayesian method of  \cite{LiZhang10}, and data-driven projection approaches such as PCR \citep{pls_functionaldata}.

Alternatively, one can use methods for ``functional predictors" \citep{ramsaysilverman, morris, reiss}, which view  time- or spatially-indexed predictors as corresponding to realizations of a function. Typically this involves estimating a time or spatially-indexed coefficient function, which is assumed to be represented as a linear combination of basis functions. Specifically, multiscale interpretations are achievable via wavelet bases. However, a discrete wavelet transform of the data will not reduce its dimensionality, making down-sampling a necessary pre-processing step for the estimation of Bayesian wavelet regression (see e.g. \citealt{vannucci}). Additionally, wavelets require the basis functions to be orthogonal; this leads to benefits in terms of identifiability and performance in estimating the individual coefficients, while also leading to disadvantages relative to ``over-complete" specifications. In fact, wavelets cannot be used when there is uncertainty on multiple pre-specified resolutions that do not conform to a wavelet decomposition (e.g. a hour/day/week structure for time-series). This is also true for ongoing research in neuroscience devoted to the development of parcellation schemes that achieve reproducible and interpretable results \citep{neuro_whichfmri, neuro_parcellations}. Researchers might be interested in using these scales jointly in a regression model to ascertain their relative contribution, but this cannot be achieved with wavelets. For classical references on wavelet regression, see \cite{donohojohnstone94, donohojohnstone95, lasso_wavelet}; from a Bayesian perspective, see e.g. \cite{vannucci} or \cite{jeong_vannucci_ko}.

To solve these issues, we propose a class of Bayesian Modular and Multiscale regression methods  (\textit{BM\&Ms}), which express the regression function as an additive expansion of functions of data at increasing resolutions. In its simplest form, the regression function becomes an additive expansion of coarse to fine step functions. This implies that multiple down-samplings of the predictor are included within a single flexible multiresolution model. Our approach can be used when (1) there is a pre-determined multiscale structure, or uncertainty on a multiplicity of pre-specified resolutions, as in the case of brain atlases; (2) with temporally- or spatially-indexed predictors, when the goal is a multiscale interpretation of single-scale data. In the first case, our method can be directly used to ascertain the contribution of the pre-determined scales to the regression function. This goal cannot be achieved via wavelets. In the second case, \bmms are related to a Haar wavelet expansion but involve a simpler, non-orthogonal transformation that facilitates easy interpretation, and suggests a straightforward extension to scalar-on-tensor regression. We address the identifiability issues induced by this non-orthogonal expansion by taking a modularization approach. The resulting \textit{BM\&M} regression is stable, well identified, easily interpretable, and provides uncertainty quantification at different resolutions.  

The idea of modularization in Bayesian inference is that instead of using a fully Bayesian joint probability model for all components of the model, one ``cuts'' the dependence between different components or modules (see \cite{modular, bettertogether} and references therein).  Modularization has been commonly applied in joint modeling contexts for computational and robustness reasons.  For example, suppose that one defines a latent factor model for multivariate predictors $X$, with the goal of using factors $\eta$ in place of $X$ in a model for the response $Y$.  Under a coherent fully Bayes model, $Y$ will impact the posterior on $\eta$.  This is conceptually appealing in providing some supervision under which one may infer latent factors that are particularly informative about $Y$.  However, it turns out that there is a lack of robustness to model misspecification, with misspecification potentially leading to inferring factors that are primarily driven by improving fit of the model and are not interpretable as summaries of $X$.  Modularization solves this problem by not allowing $Y$ to influence the posterior of $\eta$; motivated by the practical improvements attributable to such an approach, WinBUGS has incorporated a \texttt{cut(.)} function to allow cutting of dependence in routine Bayesian inference \citep{Plummer15}.

Our multiscale setting is different than previous work on Bayesian modularization, in that we use modules to combine information from the same data at increasingly higher resolutions. \cite{ChenDunson17} recently proposed a modular Bayesian approach for screening in contexts involving massive-dimensional features, but there was no multiscale structure, allowance for functional predictors, or consideration of multiple predictors that simultaneously impact the response. 

Section \ref{modular_general} introduces \bmms in a general setting, highlighting the correspondence to a data-dependent prior in a coherent Bayesian model. Section \ref{modular_linear} focuses on linear regression. 
Section \ref{comp} outlines an algorithm to sample from the modular posterior. Section \ref{data} contains applications on simulated data and to a brain imaging classification task. Proofs and technical details are included in an Appendix.

%% file: 02_modular_general_180805.tex
\section{A modular approach for multiscale regression} \label{modular_general}
 
We consider a regression problem linking a scalar response $y_i$ to an input vector $\bm{x}_i = ( x_{s,i} )_{s \in S} $ of dimension $p \gg n$, for each subject $i \in \{1, \dots, n\}$. For example, a subject's health outcome may be linked to the output from a high-resolution recording device. The goal is to predict the outcome variable and explain its variability across subjects by identifying specific patterns in the sensor recording at different resolutions. 

We denote the vector of responses by $y$ and the raw data matrix by $X_{S}$. Each row of $X_{S}$ can be down-sampled to get a new design matrix $X_{S_j}$, where $S_j$ is a lower resolution grid such that $|S_j| < |S|$. Down-sampling can be achieved by summing or averaging adjacent columns, or subsetting them. We simplify the notation slightly by calling $X = X_{S}$ and $X_j = X_{S_j}$. We consider the same data at increasing resolutions in an additive model:
\begin{align} \label{modular:joint_model}
y = f_0(X_0) + \dots + f_j(X_j) + \dots + f_K(X_K) + \varepsilon,
\end{align}
where $f_j$ is the resolution $j$ contribution to the regression function. With this additive multiresolution expansion, it is difficult to disambiguate the impact of the coarse scales from the finer ones, leading to identifiability issues. One may be able to fit $y$ using only a fine scale component, with the coarse scales discarded.  If we were to attempt fully Bayes inferences by placing priors on the different component functions, large posterior dependence would lead to substantial computational problems (e.g. very poor mixing of Markov chain Monte Carlo algorithms) and it would not be possible to reliably interpret the different $f_j$s. This happens in particular if each $f_j$ is linear, as seen in Section \ref{modular_linear}. 

We bypass these problems by adopting a modular approach \citep{modular, bettertogether}, splitting the overall joint model (\ref{modular:joint_model}) into components or \textit{modules} which are kept partly separate from each other, to get a \textit{modular posterior}. In the following, for $j > i$ we use the notation $A_{i:j} = \{ A_i, \dots, A_j \}$.

\begin{definition}
	Within the overall model (\ref{modular:joint_model}), \textit{module j} for data at resolution $S_j$ consists of a prior for $f_j$, and a model for $y | f_{1:j}$:
	\[
	 f_j | f_{1:j-1} \sim \pi(f_j| f_{1:j-1}) \qquad
	 y | f_j, f_{1:j-1} \sim p_j(y | f_j, f_{1:j-1}) \]
	where model $p_j$ estimates $f_j$ via $ y = f_1(X_1) + \dots + f_{j-1}(X_{j-1}) + f_j (X_j) + \varepsilon_j$, and $f_1, \dots, f_{j-1}$ are considered known. The output from module $j$ is the (conditional) posterior for $f_j$ obtained by prior $\pi(f_j)$ and model $p_j$, and we denote it by $m(f_j | f_{1:j-1})$:
	\begin{align} \label{bmms_general:module}
	m(f_j | f_{1:j-1}) = \frac{\pi(f_j | f_{1:j-1}) p_j(y | f_{1:j}, X_{1:j})}{p_j(y | f_{1:j-1}, X_{1:j})}.
	\end{align}
\end{definition}
Thus for the full model (\ref{modular:joint_model}) we build $K$ modules, using increasingly higher resolution data, with each module being a \textit{refinement} on previous output. 
\begin{definition} \label{bmms_general:posterior}
	The modular prior distribution for  $\bm{f} = (f_1, \dots, f_K)$ corresponds to $
	p_M( \bm{f} ) = \pi(f_1) \cdots \pi(f_j | f_{1:j-1}) \cdots \pi(f_K | f_{1:K-1})$, whereas the modular posterior distribution $p_M( \bm{f} )$ is: 
	\begin{align*}
	p_M( \bm{f} | y, X_{1:K} ) &= m(f_1) m(f_2 | f_1) \cdots m(f_K | f_{1:K-1}),
	\end{align*} 
\end{definition}
\noindent thus collecting the posteriors in (\ref{bmms_general:module}). Each module refines the output from previous modules by using higher resolution data, and the modular posterior is obtained by aggregating all refinements. 
Modularity is evidenced by resolution dependence, which is only allowed downwardly, as opposed to letting it be bidirectional as in a fully Bayes approach.
\begin{proposition} 
There exists a data-dependent prior $\pi_d$ and a likelihood $p_d$ such that $ p_M(\bm{f} | y, X) \propto \pi_d(\bm{f}) p_d(y | \bm{f}, X) $. Specifically, $p_d(\cdot)$ is the likelihood corresponding to model $ y = f_1(X_1) + f_2(X_2) + \varepsilon $, and the data-dependent prior is $\pi_d(\bm{f}) = \pi(f_1) \pi(f_2|f_1) \frac{p_1(f_1 | y, X_{1})}{p_2(f_1 | y, X_{1:2})}$.
\end{proposition}
\begin{proof}
See Appendix \ref{appendix:sec2}.
\end{proof}
\noindent The above proposition implies our modular approach to multiscale regression will resemble a fully Bayesian model if $\frac{p_1(f_1 | y, X_{1})}{p_2(f_1 | y, X_{1:2})} \approx 1$, i.e. if the modules agree on the marginal posterior probability of $f_1$, the low resolution contribution to the regression function.

%% file: 03_modular_linear_180809.tex
\section{\textit{BM\&Ms} for linear regression} \label{modular_linear}
We now assume that $f_j = X_j \theta_j$ and our goal is to study $\theta_j$ for $j=1,\dots,K$ in the model
\begin{align} \label{modular_linear:eq:complete_model} y = X_1 \theta_1 + \dots + X_K\theta_K + \varepsilon_K, \end{align}
where $\varepsilon_K \sim N(0, \sigmasq_K)$. The model includes data up to resolution $S_K$. 
We also assume that $X_j = X_{j+1} L_j = X_K \mathcal{L}_j$, meaning that $L_j$ and $\mathcal{L}_j$ respectively down-sample $X_{j+1}$ and $X_K$ to $X_j$.\footnote{Thus $\mathcal{L}_j = L_{K-1} \cdots L_1$} 
This allows us to decompose $\beta_K$, the usual regression coefficient of $y$ on $X_K$ as:
\begin{equation}\label{modular_linear:multiresolution_decomp} y = X_K \left( \mathcal{L}_1\theta_1 + \dots + \mathcal{L}_K\theta_K \right) + \varepsilon_K 
= X_K \beta_K + \varepsilon_K,
\end{equation}
thus interpreting $\theta_j$ as the contribution of resolution $S_j$ to the regression function. 

However, model (\ref{modular_linear:eq:complete_model}) is ill-posed, leading to problems with standard techniques for model fitting.  In fact, the effective design matrix $X = [X_1 \cdots X_K]$ is such that $det(X'X) = 0$, as the columns of $X_j$ are linear combinations of columns of $X_{j+1}$.  One can potentially obtain a well defined posterior through an informative prior for ${\theta}$.  However, this posterior will be highly dependent on the prior, as the likelihood has many flat regions.  

We solve this problem via modularization as in section \ref{modular_general}. The modular posterior is 
\begin{align*} 
p_M( {\theta} | y, X_{1:K}) 
&=\frac{\pi(\theta_1)p_1(y|\theta_1,X_1)}{p_1(y| X_1)} \cdots \frac{\pi(\theta_K | \theta_{1:K-1})p_K(y| \theta_{1:K}, X_{1:K})}{p_K(y| \theta_{1:K-1}, X_{1:K} )},
\end{align*}
where $p_j(\cdot)$ is the likelihood for the response $y$ under $ y = \sum_{h=1}^{j} X_{h}\theta_h + \varepsilon_j$, $\varepsilon_j \sim N(0, \sigmasq_j)$, and $\theta_1, \dots, \theta_{j-1}$ are considered known.  The posterior distribution for the coarsest scale coefficient $\theta_1$ is derived treating all the finer scale coefficients as equal to zero; this makes $\theta_1$ identifiable and interpretable as producing the best coarse scale approximation to the regression function.  In defining the posterior for $\theta_2$, we then condition on $\theta_1$ and set coefficients $\theta_3,\ldots,\theta_K$ equal to zero, and so on.  Linearity allows us to write $p_j(\cdot)$ as
\[ y - \sum_{h=1}^{j-1}X_h\theta_h = X_{j}\theta_j + \varepsilon_j, \textalt{ } \varepsilon_j \sim N(0, \sigmasq_j),\]
which essentially means that we are estimating $\theta_j$ (and $\sigmasq_j$) using the residuals $e_{j-1} =  y - \sum_{h=0}^{j-1}X_h\theta_h$ from the previous modules as responses for the current module. Hence, the modular posterior for (\ref{modular_linear:eq:complete_model}) is built using simpler, well-identified single-scale models as modules. For example, we address uncertainty on two resolutions with the following model:\footnote{An overview of the recurrent notation is in Appendix \ref{notation}}
\begin{equation} \label{modular_linear:two_scale:eq}
y = X_{1} \theta_{1} +  X_{2}\theta_{2} + \varepsilon_2, \textalt{ } \varepsilon \sim N(0, \sigmasq_2 I_n),
\end{equation}
assigning priors $\theta_j \sim N(m_j, \sigmasq_j M_j)$, $\sigmasq_j \sim \text{InvG}(a,b)$ for $j\in \{1,2 \}$, and using a first module corresponding to model $y = X_1 \theta_1 + \varepsilon_1$ with $\varepsilon_1 \sim N(0, \sigmasq_1 I_n)$.
\begin{proposition}\label{modular_linear:prop_joint_posterior}
	The modular posterior of $\theta | \sigmasq_1, \sigmasq_2$ for model (\ref{modular_linear:two_scale:eq}) is $N\left( \mu_{1:2}, \Sigma_{1:2} \right)$ where
	\begin{align} \label{modular_linear:joint_posterior}
	\begin{array}{ccc}
	\mu_{1:2} = \begin{bmatrix} \mu_1 \\ \mu_2 \end{bmatrix} = \begin{bmatrix} \mu_{\beta_1} \\ \mu_{\beta_2} - Q_1 \mu_{\beta_1} \end{bmatrix} & \qquad & \Sigma_{1:2} = \begin{bmatrix} \sigmasq_1 \Sigma_1 & -\sigmasq_1 \Sigma_1 Q_1' \\ -\sigmasq_1 Q_1 \Sigma_1 & \sigmasq_2 \Sigma_2 + \sigmasq_1 Q_1 \Sigma_1 Q_1'  \end{bmatrix},
	\end{array}
	\end{align}
	with $Q_1$ = $\Sigma_2X_2'X_1$, and $\mu_{\beta_j} = \Sigma_j^{-1} (M_j m_j + X_j'y) $ for $j\in \{1,2 \}$.\footnote{Note that $\mu_{\beta_j}$ is the posterior mean we would obtain from a single resolution model of the form $
	y = X_j \beta_j + \varepsilon_j$ when we assign prior $\beta_j \sim N(m_j, \sigmasq_j M_j)$, $j \in \{1,2 \}$.} Proof in Appendix \ref{appendix:twoscale_posterior}.
\end{proposition}
\noindent Finally, note we can estimate $\beta_2$ in $y = X_2 \beta_2 + \varepsilon_2$ by accumulating all components, i.e. using $L_1 \theta_1 + \theta_2$, which reconstructs the decomposition in (\ref{modular_linear:multiresolution_decomp}).
\subsection{Asymptotics of \textit{BM\&Ms} in linear regression}\label{modular_linear:bvm}
We now assume that $(y, \mathcal{X})$ were generated according to a process such that 
\[ \frac{1}{n} \begin{bmatrix} y'y & y'\mathcal{X} \\ \mathcal{X}'y & \mathcal{X}'\mathcal{X} \end{bmatrix} \xrightarrow[a.s.]{} \begin{bmatrix} \omega_{yy} & \mathbf{\omega}_{y \mathbf{x}} \\ \mathbf{\omega}_{\mathbf{x} y} & \Omega \end{bmatrix}, \]
a positive definite matrix, and $ y | \mathcal{X}, \sigmasq \sim N(\mathcal{X} b, \sigmasq I_n)$, where $b \in \mathbb{R}^p$ with dimension not depending on the sample size. We consider a two-scale linear model like in (\ref{modular_linear:two_scale:eq}). We assume $X_2 = \mathcal{X}\mathcal{L}_2$, $X_1 = X_2 L_1$, and assign priors on $\theta_j$ that have positive density on $\mathbb{R}^{p_j}$. Finally, we assume $\sigmasq_j$ is known for each module.
\begin{proposition}
	The modular posterior distribution of $(\theta_1, \theta_2)$ in model $ y = X_1 \theta_1 + X_2 \theta_2 + \varepsilon $ is approximated by $N(\bar{\mu}_{1:2}, \bar{\Sigma}_{1:2} )$, where
\begin{equation*} \label{modular_linear:asymptotic_dist}
\begin{aligned}
\begin{array}{ccc}
\bar{\mu}_{1:2} = 
\begin{bmatrix}
\theta_1^* \\
\theta_2^*
\end{bmatrix} = \begin{bmatrix}
\beta_1^* \\
\beta_2^* - L_1 \beta_1^*
\end{bmatrix} & \quad & 
\bar{\Sigma}_{1:2} =
\begin{bmatrix}
\sigmasq_1\Omega_1^{-1} & - \sigmasq_1 \Omega_1^{-1}L_1'\\
- \sigmasq_1 L_1 \Omega_1^{-1} & \sigmasq_2 \Omega_2^{-1} +  \sigmasq_1 L_1 \Omega_1^{-1} L_1'
\end{bmatrix}
\end{array}
\end{aligned}
\end{equation*}
where $\beta_j^*$ is the pseudo-true value of $b$ under the model $y = X_j \beta_j + \varepsilon$.
	\begin{proof}
		See Appendix \ref{appendix:bvm}.
	\end{proof}
\end{proposition}

\noindent This is also the asymptotic distribution of $(\hat{\theta}_1, \hat{\theta}_2)$, where $\hat{\theta}_1 = \hat{\beta}_1$ and $\hat{\theta}_2$ is the least-squares estimator obtained by regressing $y-X_1 \hat{\beta}_1$ on $X_2$. Hence, in large samples, \bmms correspond to a sequential least squares procedure that regresses the residuals of coarser models on higher resolution data. Our approach propagates uncertainty across multiple stages on finite samples, unlike many two-stage procedures (see e.g. \cite{murphytopel} for a treatment of two-stage estimators in econometrics). 
\begin{corollary} \label{bvm:corollary1}
	The large sample distributions of $\theta_1$ and $L_1\theta_1 + \theta_2$ are approximated by $N(\beta_1^*, \frac{\sigmasq_2}{n} \Omega_1^{-1})$ and $N(\beta_2^*, \frac{\sigmasq_2}{n} \Omega_2^{-1})$, respectively. In other words, accumulating the modular posterior mean components up to $j$ results in a consistent and asymptotically efficient estimator for $\beta_j^*$.
\end{corollary}
\begin{proof}
	This is a direct consequence of the properties of the Normal distribution.
\end{proof}

%% file: 06_computation_180810.tex
\section{Computation of the modular posterior} \label{comp}
We sample from the modular posterior of \ref{bmms_general:posterior} by sequentially sampling from each module. 

\vspace{0.8cm}
\begin{algorithm}[H] \label{comp:algorithm_fix}
	\SetKwInput{KwStart}{Start}
	\SetKwInput{KwDo}{Do}
	\SetKwInput{KwFinally}{Finally}
	\SetKwInput{KwInput}{Input}
	\SetKwInput{KwOutput}{Output}
	
	\For{$t \in \{1, \dots, T \}$}{
		\KwStart{Draw sample $f_{1}^{(t)}$ from the module $1$ posterior $m(f_{1}) = p_{1}(f_1 | y, X_{1})$ }
		\For{$j \in \{2, \dots, K\}$}{
				Draw sample $f_{j}^{(t)}$ from the module $j$ posterior \begin{equation}\label{comp:cond_post} m(f_{j} | f_{1:j-1}^{(t)} ) = p_j(f_j | f_{1:j-1}^{(t)}, y, X_{1:j}) = p_M(f_j | \bm{f}_{-j}, y, X_{1:K}) \end{equation}
		}
	}
	\vspace{0.5cm}
	\caption{Sampling $\left\{ \bm{f}^{(t)} \right\}_{t \in \{1, \dots, T \}}$ from the modular posterior $p_M( \bm{f} | y, X_K )$}
\end{algorithm}
\vspace{0.4cm}
\noindent Obtaining a sample from the modular posterior depends on $m(f_{j} | f_{1:j-1}^{(t)} )$, which is determined by the module choice.  Therefore, in a multiscale linear regression with conjugate priors as in section \ref{modular_linear} we can easily sample from each individual module taking advantage of Eq. (\ref{modular_linear:joint_posterior}). In the more complex cases where $m(f_{j} | f_{j-1}^{(t)} )$ is approximated via MCMC, we can use the fact that for all $j$, $p_j(f_j | f_{1:j-1}, y, X_{1:j-1}) = p_M(f_j | f_{-j}, y, X_{1:K})$, i.e. the posterior distribution of each module corresponds to the full conditional distribution of the overall modular model, hence sampling from a module's posterior can be seen as a ``modular" Gibbs step. This also means that the computational complexity of \bmms is of the same order of magnitude of each component module, since $K$ is taken to be small.

%% file: 07_data_analysis.tex
\vspace{-0.25cm}
\section{Applications}\label{data} 
\vspace{-0.25cm}
\subsection{Simulation study} \label{data:simulated}
We generate $n=60$ observations from a linear regression model $ y = X \beta + \varepsilon$, with
\begin{align*}
\begin{matrix}
 \varepsilon \sim N(0, I_n) & \quad & \Omega = (\omega_{h,j})  \quad h,j = 1, \dots, p  & \quad & p = 128 & \\
 \mathbf{x}_i \sim N(0, \Omega)   & \quad & \omega _{h,j} = \exp \{ - (1-\rho) | h-j | \}  & \quad & \rho = 0.98
\end{matrix}
\end{align*}
and where $\beta$ is a $p \times 1$ vector obtained by discretizing the \textit{Doppler}, \textit{Blocks}, \textit{HeaviSine}, \textit{Bumps}, and \textit{Piecewise Polynomial} functions of \cite{donohojohnstone94, donohojohnstone95, nason_silverman} on a regular grid. Notice that covariates are highly correlated, and the sample size is relatively low, suggesting that the regression function may be challenging to estimate at a fine resolution. Our goals are thus: (1) dimension reduction and multiscale decomposition of $\beta$; (2) estimation and uncertainty quantification of the relative contributions of different scales; (3) out-of-sample prediction. The $k^{\text{th}}$ module, given $\theta_{k-1}$, is specified by $y_k = X\theta_k + \varepsilon_k$, with $\varepsilon_k \sim N(0, \sigmasq_k)$, $y_k = y_{k-1}-X\theta_{k-1}$ and $\theta_k = (\theta_{k, 1}, \cdots, \theta_{k, p} )'$ with
\begin{align} \label{data:changepoint}
\theta_{k, j} &= \begin{cases}
b_1, & \text{if}\ 0 < j \leq t_1 \\
\dots \\
b_{H_k}, & \text{if}\ t_{H_k-1} < j \leq t_{H_k}=p
\end{cases}
\end{align}
the goal being estimation of $b_1, \dots, b_{H_k}$ and the split locations $t_1, \dots, t_{H-1}$. We use 3 modules and fix $H_1 = 1, H_2=2, H_3=4$. This  allows us to estimate an unknown, hierarchical, 3-scale decomposition of $\beta$. MCMC approximations of the modular posterior decomposition of $\beta$ in the \textit{Heavi} and \textit{Blocks} cases are in Figure \ref{data:simulation:interpret}. 
\begin{figure}[H] 
	\centering
	\includegraphics[width=0.75\textwidth]{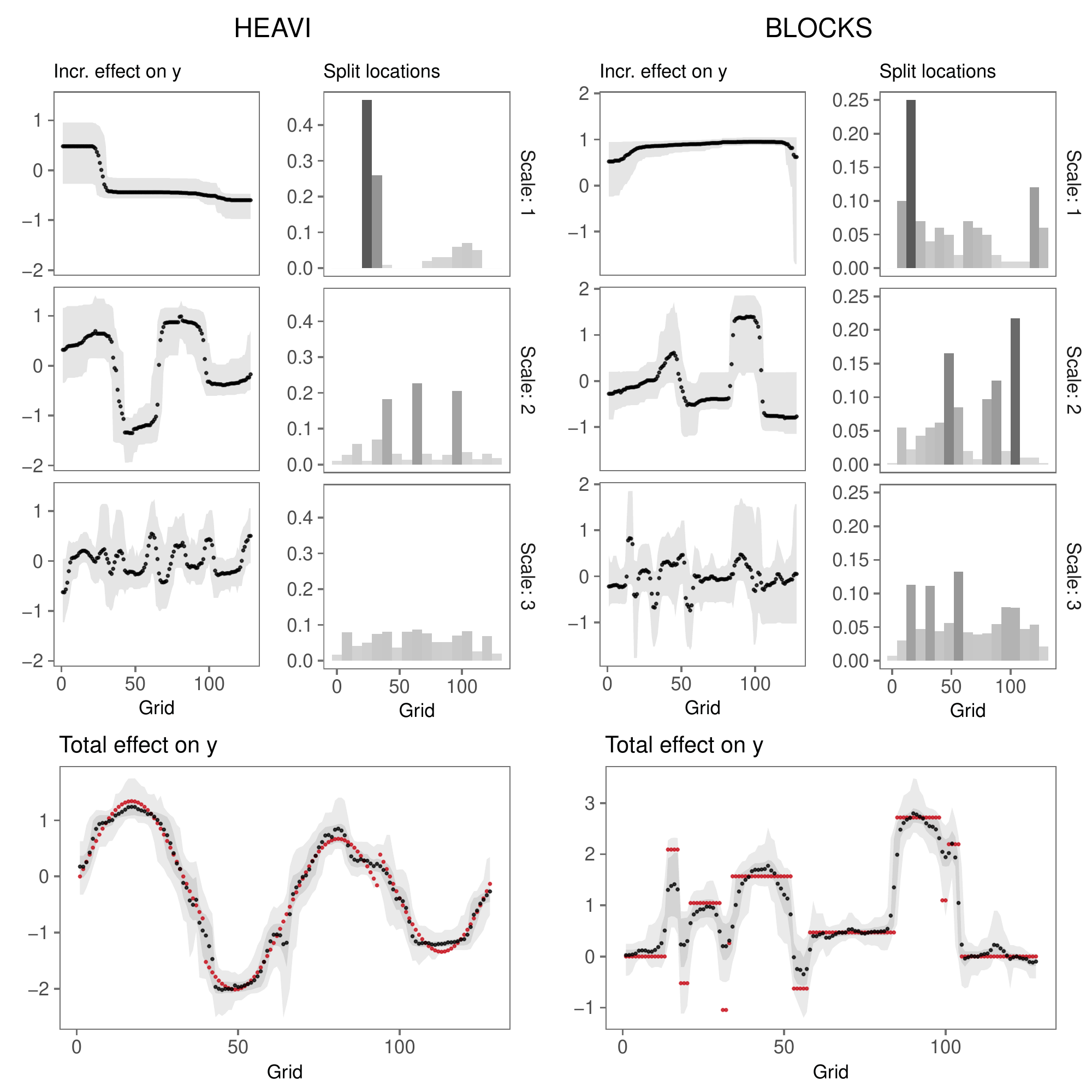}
	\caption{Model-averaged \bmms decomposition of $\beta$. The total effect is the sum of the intermediate scale contributions. In red, the true unknown $\beta$. For \textit{Heavi}, Scale 3 has little to no contribution to the overall estimation, as evidenced by the 95\% credible bands including 0. For \textit{Blocks}, Scale 3 refines on the previous ones in a few select locations.}\label{data:simulation:interpret} 
\end{figure}
\vspace{-0.5cm}
Within the above setup, \bmms compare favorably to the other models in almost all cases (Figure \ref{data:simulation:perform}), owing to an (implicit) multiscale structure (\textit{Doppler}, \textit{Blocks}), and/or non-smoothness (\textit{Blocks}, \textit{Bumps}).\footnote{Details on each implemented model are in Appendix \ref{appendix:key_models}.}
\begin{figure}[H] 
	\centering
	\includegraphics[width=\textwidth]{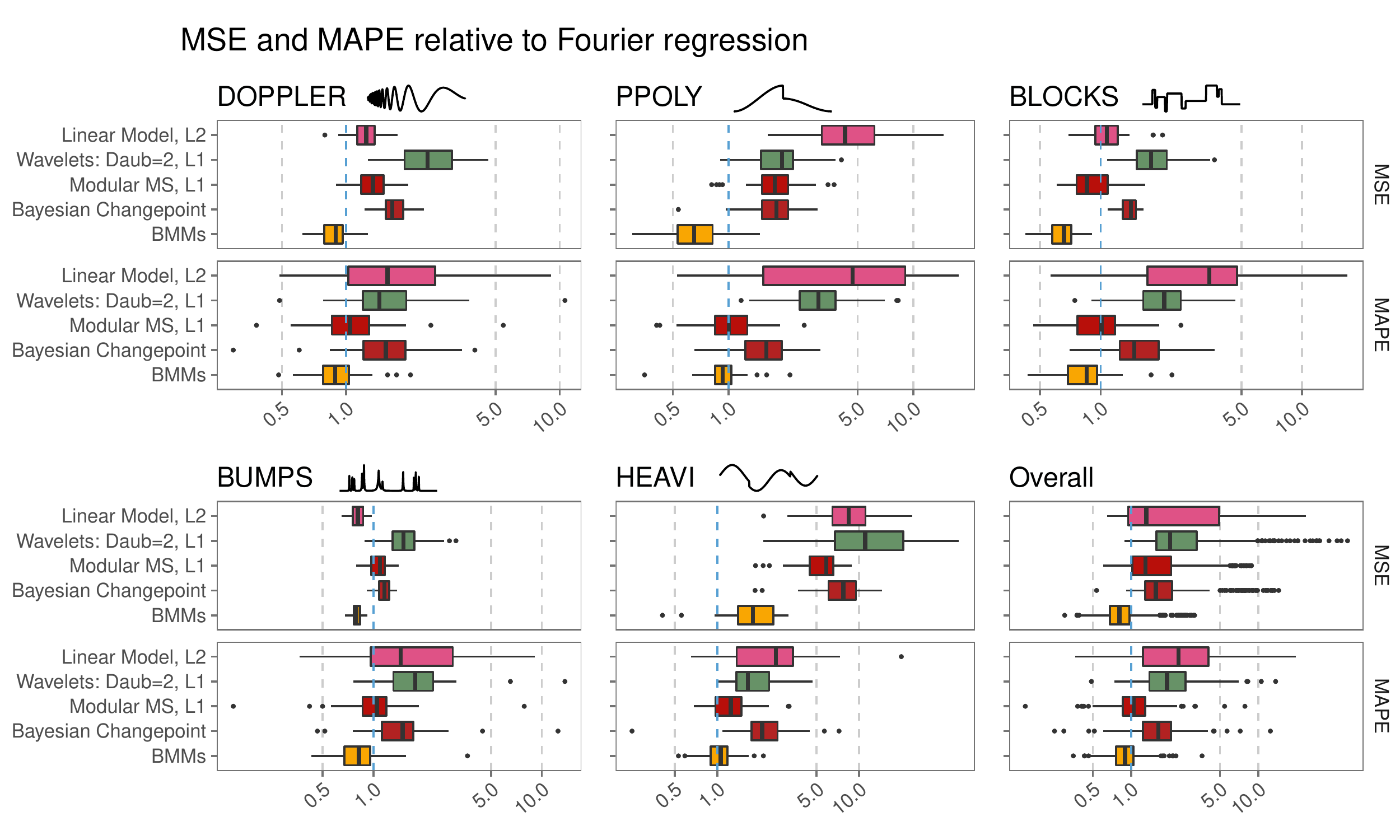}
	\caption{Mean Square Error in estimating $\beta$, and Mean Absolute Prediction Error on 100 samples of size $n_{\text{out}}=100$, relative to Fourier regression on the same data (blue dashed line).}\label{data:simulation:perform} 
\end{figure}
\vspace{-0.5cm}

\subsection{Gender differences in multiresolution tfMRI data} \label{data:brain}
Brain activity and connectivity data plays a central role in neuroscience research today, but increasingly higher-resolution medical imaging devices make management and analysis of such data challenging. We use data from the \textit{Human Connectome Project} (HCP) \citep{hcp_overview, hcp_pre1}, considering a sample of $n=100$ subjects, with the goal of classifying subjects' genders using brain activation data during task-based functional Magnetic Resonance Imaging (tfMRI). Gender differences have been observed in neuroscience and linked to brain morphology and connectivity \citep{hcp_gender2, hcp_gender4}, or task-based activity patterns \citep{hcp_gender1, hcp_gender3}. For an overview see \citet[Chapter~7]{neurosci101}. 
We use \bmms on tfMRI data recorded during the \textit{social} task in HCP, preprocessed according to the \textit{Gordon333} \citep{gordon333} parcellation. This hierarchical partitioning splits the brain into a multiscale structure of 333 regions, 26 lobes, and 2 hemispheres.
 
While we expect each region to have low explanatory power on its own, it is unclear whether grouping them to form a coarser structure might improve predictive power. In particular, while the coarse-scale 26 lobes are easily interpretable given their connection to known brain functions, they might not be an efficient coarsening for our predictive task. We thus implement \bmms in two ways, following the two multiscale points of view of Section \ref{bmms:intro}: (1) we consider the regions-lobes multiscale structure as specified by \cite{gordon333}; (2) we use regions' centroid information to adaptively collapse them into coarser groups. 
In both cases, we consider a binary regression model with probit link, and hence assume that $Pr(y_i=1) = \Phi \left( \mathbf{x}_{1, i} \theta_{1} + \dots + \mathbf{x}_{K, i} \theta_{K} \right)$, where $y_i$ is binary and corresponds to the subject's gender, and $\mathbf{x}_{j, i}$ is the same subject's data at resolution $j$. 

We adapt the Gibbs sampling algorithm of \cite{albertchib93} to \bmms. We thus introduce a latent variable $Z_i$ for each subject, and alternate sampling from $p(Z_i|y_i, \theta_{1:K}, \mathbf{x}_{i, 1:K})$ and $p_M(\theta_{1:K}|y, Z, X_{1:K})$. Here, $p(Z_i|y_i, \theta_{1:K}, \mathbf{x}_{i, 1:K})$ is the usual truncated Normal distribution, centered at $\mathbf{x}_{i, 1}\theta_{1} + \dots + \mathbf{x}_{i, K}\theta_{K}$, with unit variance, and truncation on zero to the left if $y_i=1$, to the right if $y_i=0$. On the other hand, $p_M$ denotes the modular posterior of a linear model as above.

In implementation (1), $X_{j}$ is a $n \times p_j$ matrix collecting subjects' brain activation data at resolutions $j\in \{1,2\}$ corresponding to regions and lobes, so that $p_1=333, p_2=26$. We use Bayesian Spike \& Slab modules to estimate $\theta_{j} = (\theta_{j,1} , \cdots , \theta_{j, p_j})$ for $j \in \{1, 2\}$. We are interested in estimating the posterior distributions of $\theta_{j}$, and to understand whether lobes provide additional information compared to the baseline of regions. The estimated posterior means are shown at the top of Figure \ref{data:brain:bmms}. In implementation (2), we fix $K=5$ to decompose the original higher-resolution parcellation. Implementation of each module follows the lines of (\ref{data:changepoint}), but using two-dimensional spatial information, as detailed in Appendix \ref{appendix:soi_reg}. The estimated posterior means are shown at the bottom of figure \ref{data:brain:bmms}. 
\begin{figure}
	\centering
	\includegraphics[width=0.9\textwidth]{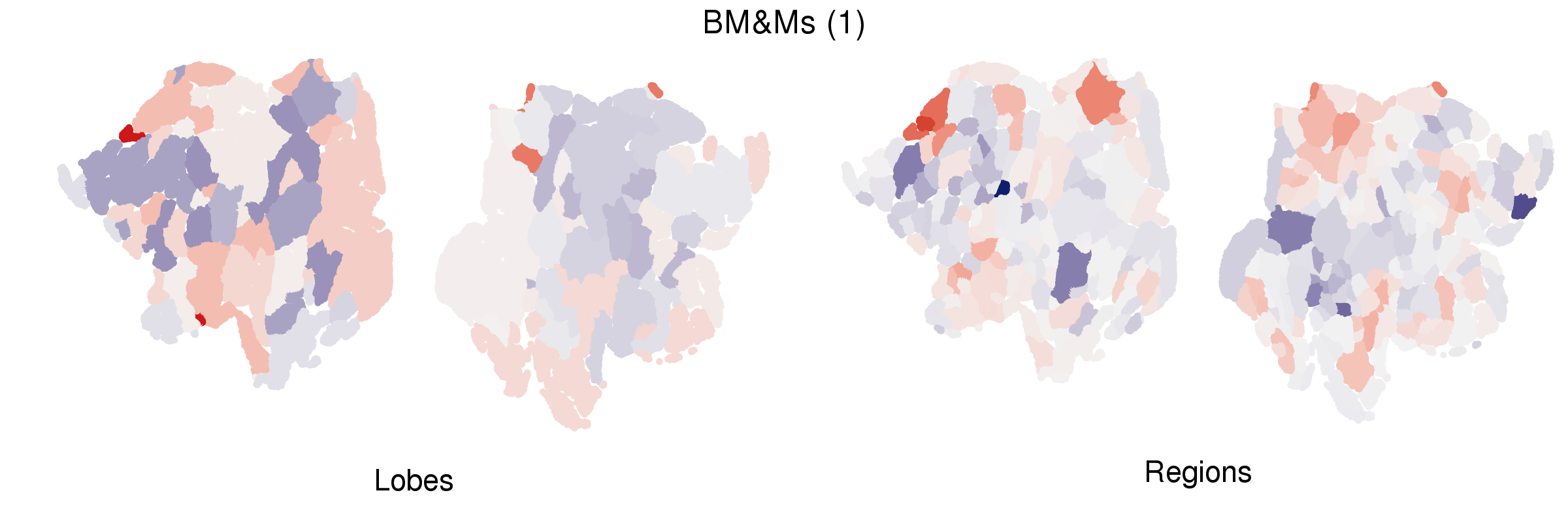}
	\includegraphics[width=0.85\textwidth]{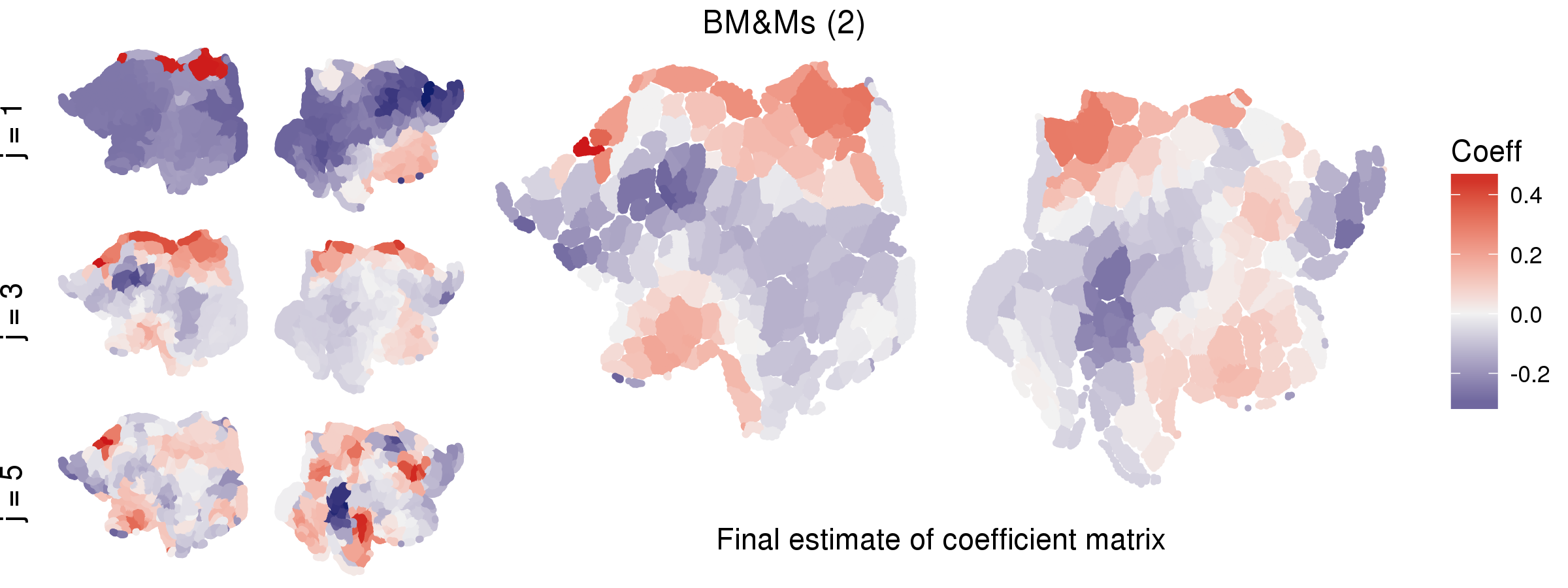}
	\caption{Posterior means estimated with \bmms. Top: lobes-regions multiscale structure. Bottom: Multiscale decomposition (left) of an estimated coefficient image (right) in a scalar-on-image regression setting.}\label{data:brain:bmms} 
\end{figure}
\begin{figure}
	\centering
	\includegraphics[width=0.85\textwidth]{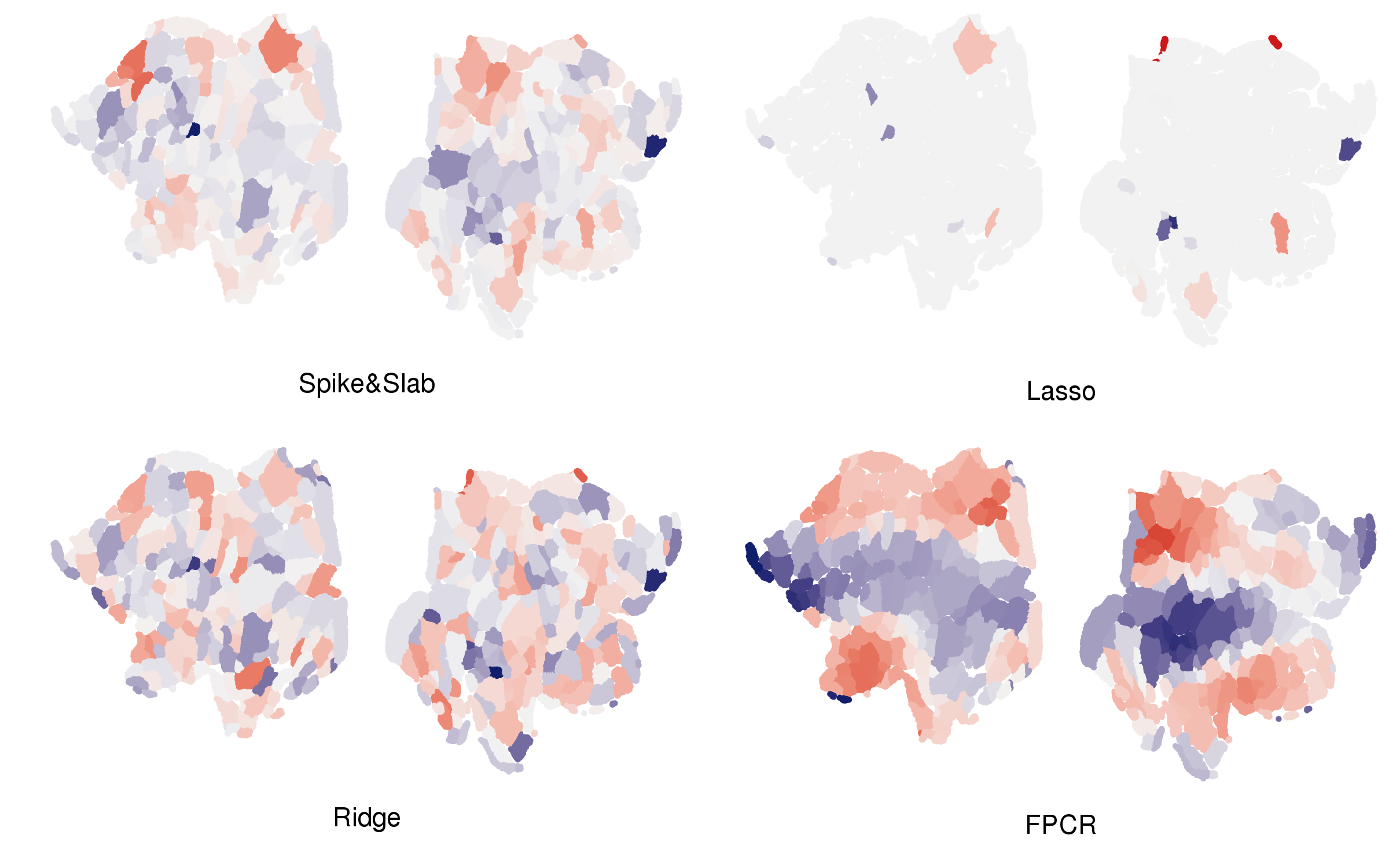}
	\caption{Coefficient on brain regions estimated via four alternative single-scale models.}\label{data:brain:others} 
\end{figure}
\begin{table}
	\small
	\centering
	\begin{tabular}{lrrccc}
	  \hline
	  Model & Accuracy & AUC & Multiscale & Bayesian & Scalar-on-Image \\ 
	  \hline
	 \bmms (1) & 0.806 & \textbf{0.889} & \checkmark & \checkmark &   \\ %
	 \bmms (2) & \textbf{0.810} & \textbf{0.889} & \checkmark & \checkmark & \checkmark  \\
	 Spike \& Slab & 0.800 & 0.883 & & \checkmark &   \\ %
	   Ridge & 0.792 & 0.871 & & &  \\
	   Lasso & 0.749 & 0.824 & & &  \\ 
	   Functional PCR & 0.781 & 0.856 & & & \checkmark  \\
	   \hline
	   \normalsize
\end{tabular}
\caption{Correct classification rate (\textit{Accuracy}), and area under ROC curve (\textit{AUC}), on random samples of size $n_{\text{out}} = 385$, averaged across 200 resamples of the data.}\label{data:brain:perform}
\end{table}

We compare \bmms to competing models on the same data. Table \ref{data:brain:perform} reports the out-of-sample performance of all tested models, whereas Figure \ref{data:brain:others} shows the estimation output of the competitors.\footnote{Refer to Appendix \ref{appendix:key_models} for details on the implemented models.} \bmms perform at least as well as competing models, while providing additional information on the multiscale structure. On one hand, ridge regression achieves good out-of-sample performance but the resulting estimates prevent a clear understanding of the results. Lasso regression would lead researchers to believe some regions might account for gender differences, yet its performance is the worst among tested models. Spike \& Slab priors result in many regions being infrequently selected, and this may lead researchers to believe that coarser scale information from lobes might be useful. However, applying \bmms on the regions-lobes multiscale structure results in little to no improvement in predictive power, with lobes being selected rarely. Finally, when considering the spatial structure of the brain regions, the estimated coefficient images of \bmms and FPCR are similar. However, the \bmms image can be decomposed in coarse-to-fine contributions (shown in Figure \ref{data:brain:bmms}, bottom left), which aid in the interpretation of the final estimate.

%% file: 08_discussion.tex
\section{Discussion}
In this article, we introduced a Bayesian modular approach that builds an overall model by the sequential application of increasingly more complex component modules. Our approach can be applied to multiscale regression problems in two common scenarios: (1) when multiple resolutions of the data could be used to model the regression relationship, to assess their contribution to the regression function; (2) when the focus of analysis is on a multiscale interpretation of results. Compared to established methods for multiscale regression such as wavelets, our method is more flexible and with the potential of easier interpretation. Both simulations and real data analysis show that this is not achieved at the expense of performance.

The modular posterior of \bmms is the product of each module's posterior. This implies that our method inherits properties of the chosen component modules. Choosing component modules requires clarity on what the objective of analysis is. For example, we showed in Section \ref{data} that we can use variable selection modules when the resolution structure is pre-specified, and changepoint-detection modules to find a multiscale interpretation. However, other module choices can be explored for different objectives, and overall models with mixed-type modules might offer additional interpretation opportunities. Additionally, our approach can also be used in non-parametric regression settings by considering the identity matrix $I_n$ as the high-resolution design matrix.

%% file: 91_appendix_180810.tex
\section{Notation}\label{notation}
\begin{table}[H]
\centering
\begin{tabular}{ll}
  \hline
 $y$ & output vector of size $n \times 1 $ \\
 $S_j$ & resolution of the data, $j \in \{ 1, \dots, K \}$ \\
 $\mathcal{S}$ & unknown ``true" resolution \\
 $X_{S_j} = X_j$ & design matrix at resolution $S_j$, size $n \times p_j$ \\
 $\mathcal{X}$ & data matrix corresponding to $\mathcal{S}$ \\
 $b$ & ``true" coefficient vector, such that $y = \mathcal{X}b + \varepsilon$ \\
 $\mathcal{L}_j$ & $p \times p_j$ matrix such that $\mathcal{X} \mathcal{L}_j = X_j$ \\
 $L_j$ & $p_{j+1} \times p_{j}$ matrix such that $X_{j+1} L_j = X_j$ \\
 $\beta_j$ & coefficient vector in single-resolution model $y = X_j \beta_j + \varepsilon_j$ \\
  $\mu_{\beta_j}$ & posterior mean for $\beta_j$ in the single-resolution model $y = X_j \beta_j + \varepsilon_j$ \\
  $\hat{\beta}_j$ & LS estimate of $\beta_j$  in the single-resolution model $y = X_j \beta_j + \varepsilon_j$ \\
 $\sigmasq_j$ & variance of $\varepsilon_j$ \\
 $L_j$ & coarsening operator such that $X_{j+1} L_j = X_j$ \\
 $\mathcal{L}_j$ & coarsening operator such that $X_K \mathcal{L}_j = X_j$ \\
 $\theta_j$ & $j$th component of the multiresolution parameter $\theta$ \\
 $\theta = (\theta_1 \cdots \theta_K)$ & multiresolution parameter for a K-resolution model \\
 $\beta_K = \mathcal{L}_1 \theta_1 + \dots + \mathcal{L}_K \theta_K $ & multiresolution decomposition of the coefficient vector \\
 $N(m_j, \sigmasq_j M_j)$ & prior for $\theta_j$ \\
 $N(\mu_j, \sigmasq_j \Sigma_j)$ & posterior for $\theta_j$ obtained with module $j$\\
   \hline
\end{tabular}
\end{table}

\section{Appendix to Section 2} \label{appendix:sec2}
\begin{proposition} 
	There exists a data-dependent prior $\pi_d$ and a likelihood $p_d$ such that $ p_M(\bm{f} | y, X) \propto \pi_d(\bm{f}) p_d(y | \bm{f}, X) $. Specifically, $p_d(\cdot)$ is the likelihood corresponding to model $ y = f_1(X_1) + f_2(X_2) + \varepsilon $, and the data-dependent prior is $\pi_d(\bm{f}) = \pi(f_1) \pi(f_2|f_1) \frac{p_1(f_1 | y, X_{1})}{p_2(f_1 | y, X_{1:2})}$.
\end{proposition}
\begin{proof}
We consider without loss of generality the case with two modules represented by data $X_1$ and $X_2$, respectively. Here,  $p_1(\cdot)$ corresponds to the reduced model $y = f_1(X_1) + \varepsilon_1$, while $p_2(\cdot)$ to $ y = f_1(X_1) + f_2(X_2) + \varepsilon $. 
We build the modular posterior using the modules' posteriors, and multiply and divide by $p_2(y | X_{1:2})$:
\begin{align*}
p_M( \bm{f} | y, X ) &= m(f_1) m(f_2 | f_1) = \frac{\pi(f_1) p_1(y | f_1, X_1)}{p_1(y | X_1)} \frac{\pi(f_2|f_1) p_2(y | f_{1}, f_2, X_{1:2})}{p_2(y | f_{1}, X_{1:2})} \\
&= \pi(f_1) \pi(f_2|f_1) \frac{p_1(y | f_1, X_1) p_2(y | X_{1:2})}{p_2(y | f_{1}, X_{1:2}) p_1(y | X_1)} \frac{p_2(y | f_{1}, f_2, X_{1:2})}{p_2(y | X_{1:2})}.
\end{align*}
We thus interpret $\frac{p_2(y | f_{1}, f_2, X_{1:2})}{p_2(y | X_{1:2})}$ as the likelihood divided by the normalizing constant. This means the modular posterior can be obtained via the full model $y = f_1(X_1) + f_2(X_2) + \varepsilon$ where both $f_1$ and $f_2$ are unknown, and a prior on $f_1, f_2$ which depends on the data through the following factor:
\begin{align*}
& \frac{p_1(y | f_1, X_{1})p_2(y| X_{1:2})}{p_2(y | f_1, X_{1:2})p_1(y | X_{1})} = \frac{p_1(y | f_1, X_1)}{p_1(y | X_1)} \cdot \frac{p_2(y | X_{1:2})}{p_2(y | f_1, X_{1:2})} = \frac{p_1(f_1 | y, X_1)}{\pi_1(f_1)} \cdot \frac{p_2(y | X_{1:2})}{p_2(y | f_1, X_{1:2})}
\end{align*}
We then use the fact that according to model $p_2$,
\begin{align*}
p_2(y | X_{1:2}) &= \frac{\pi(f_1) \pi(f_2|f_1) p_2(y | f_1, f_2, X_{1:2}) }{p_2(f_1, f_2 | y, X_{1:2})} = \frac{\pi(f_1) \pi(f_2|f_1) p_2(y | f_1, f_2, X_{1:2}) }{p_2(f_2 | y, X_{1:2}, f_1) p_2(f_1 | y, X_{1:2})}.
\end{align*}
and hence we obtain that
\begin{align*}
&\frac{p_2(y | X_{1:2})}{p_2(y | f_1, X_{1:2})} = \frac{\pi(f_1)}{ p_2(f_1 | y, X_{1:2})} \cdot  \frac{\pi(f_2|f_1) p_2(y | f_1, f_2, X_{1:2})}{p_2(y | f_1, X_{1:2}) p_2(f_2 | y, X_{1:2}, f_1)}.
\end{align*}
Since $p_2(f_2 | y, X_{1:2}, f_1) = \frac{\pi(f_2|f_1) p_2(y | f_1, f_2, X_{1:2})}{p_2(y | f_1, X_{1:2})}$, we can write the modular posterior as 
\begin{align*}
p_M( \bm{f} | y, X ) &= \pi(f_1) \pi(f_2 | f_1) \frac{p_1(f_1 | y, X_{1})}{p_2(f_1 | y, X_{1:2})} \cdot \frac{p_2(y | f_{1:2}, X_{1:2})}{p_2(y|X_{1:2})}.
\end{align*}
\end{proof}

\section{Appendix to Section 3}

\subsection{Two-scale \textit{BM\&Ms} posterior} \label{appendix:twoscale_posterior}

The modular posterior of $\theta | \sigmasq_1, \sigmasq_2$ for the \textit{BM\&Ms} of model (\ref{modular_linear:two_scale:eq}) is
\begin{align} 
\theta = \begin{bmatrix}
\theta_1 & \theta_2
\end{bmatrix}' \sim N\left( \mu_{1:2}, \Sigma_{1:2} \right)
\end{align}
\begin{align*}
\begin{array}{ccc}
\mu_{1:2} = \begin{bmatrix} \mu_1 \\ \mu_2 \end{bmatrix} = \begin{bmatrix} \mu_{\beta_1} \\ \mu_{\beta_2} - Q_1 \mu_{\beta_1} \end{bmatrix} & \qquad & \Sigma_{1:2} = \begin{bmatrix} \sigmasq_1 \Sigma_1 & -\sigmasq_1 \Sigma_1 Q_1' \\ -\sigmasq_1 Q_1 \Sigma_1 & \sigmasq_2 \Sigma_2 + \sigmasq_1 Q_1 \Sigma_1 Q_1'  \end{bmatrix}
\end{array}
\end{align*}
where we denote $Q_1$ = $\Sigma_2X_2'X_1$, and $\mu_{\beta_j}$ with $j\in \{0,1 \}$ are the posterior means we would obtain from single resolution models of the form 
\begin{align}
y = X_j \beta_j + \varepsilon_j,
\end{align} 
when we assign prior $\beta_j \sim N(m_j, \sigmasq_j M_j)$.

We build \textbf{module 1}, $m(\theta_1)$ with a prior $\pi(\theta_1, \sigmasq_1)$ and a model $p_1(y | \theta_1, \sigmasq_1)$:
	\begin{align*}
	\pi(\sigma^2_1) \propto \frac{1}{\sigmasq_1} &\textalt{ }  \pi(\theta_{1} | \sigmasq_1) = N(m_1, \sigmasq_1 M_1) 
	\end{align*}
	\[ p_1(y | \theta_1, \sigmasq_1, X_1) = N(X_1 \theta_1, \sigmasq_1 I_n) \]
	From this we obtain the posterior
	\begin{align*}
	\begin{array}{lcl}
	p(\theta_{1} | \sigma^2_1, y, X_1) = N(\mu_1, \sigma^2_1 \Sigma_1) & \textalt{where} &
	\Sigma_1 = \left(M_1^{-1} + X_{1}'X_{1}\right)^{-1}\\  
	& \quad &
	\mu_1 = \Sigma_1 \left(M_1^{-1}m_1 + X_{1}'y \right)
	\end{array}
	\end{align*}
	Conditioning on $\theta_{1} = \tilde{\theta}_{1}$, we then build \textbf{module 1}:
	\begin{align*}
	\pi(\sigma^2_2) \propto \frac{1}{\sigmasq_2} &\textalt{ } 
	\pi(\theta_{1} | \sigmasq_2) = N(m_2, \sigmasq_2 M_2)
	\end{align*}
	\[
	p_2(y | \theta_1=\tilde{\theta}_1, \sigmasq_2, X_{1:2}) = N(X_1\tilde{\theta}_1 + X_2 \theta_2, \sigmasq_2 I_n) \]
	Hence, the posterior for $\theta_{1}|\theta_1$ will be 
	\begin{align*}
	\begin{array}{lcl}
	p(\theta_{1} | \theta_{1}=\tilde{\theta}_1, \sigma^2_2, y, X_{1:2}) = N(\mu_2, \sigmasq_2 \Sigma_2) & \textalt{where} &
	\Sigma_2 = \left(M_2^{-1} + X_{1}'X_{1}\right)^{-1}\\  
	& \quad &
	\mu_2 = \Sigma_2 \left(M_2^{-1}m_2 + X_{1}'(y - X_{1}\tilde{\theta}_{1}) \right)
	\end{array}
	\end{align*}
	We find the posterior distribution $\theta$ via modules 0 and 1:
	\begin{align*}
	\begin{array}{ccc}
	\theta_1 | \sigmasq_1, y, X_1 \sim N( \mu_1, \sigmasq_1 \Sigma_1) & \qquad & 
	\theta_2 | \theta_1, \sigmasq_1, \sigmasq_2, y, X_2 \sim N( \mu_2, \sigmasq_2 \Sigma_2 )
	\end{array}
	\end{align*}
	The end result follows from the properties of the Normal distribution.

\subsection{Asymptotics for \textit{BM\&Ms}} \label{appendix:bvm}
The goal of this section is to show that the modular posterior in linear models is approximately normal in large samples, with a mean that is a composition of the (rescaled) pseudo-true regression coefficients at different resolutions. In order to do so, it will be sufficient to show that each module results in approximately normal (conditional) posteriors. This will ensure normality of the overall modular posterior, by the properties of the normal distribution.

We consider response vector $y$ and data matrix $\mathcal{X}$, and following \cite{geweke2005} we assume that $(y, \mathcal{X})$ were generated according to a process such that 
\[ \frac{1}{n} \begin{bmatrix} y'y & y'\mathcal{X} \\ \mathcal{X}'y & \mathcal{X}'\mathcal{X} \end{bmatrix} \xrightarrow[a.s.]{} \begin{bmatrix} \omega_{yy} & \mathbf{\omega}_{y \mathbf{x}} \\ \mathbf{\omega}_{\mathbf{x} y} & \Omega \end{bmatrix}, \]
a positive definite matrix, and
\[ y | \mathcal{X}, \sigmasq \sim N(\mathcal{X} b, \sigmasq I_n), \] 
where $b \in \mathbb{R}^p$ with dimension not depending on the sample size, and $\sigmasq$ is known. We consider a finite set of predetermined resolutions $S_1, \dots, S_K$, corresponding to $X_1, \dots, X_K$, respectively, such that $X_{j} = \mathcal{X} \mathcal{L}_j$ and $X_{j+1}L_j = X_j$ for some $\mathcal{L}_j$ and $L_j$, $j \in \{1, \dots, K \}$. Here, $\mathcal{L}_j$ and $L_j$ are coarsening operators that perform partial sums of adjacent columns and hence reduce the data dimensionality to $p_j$, from $p$ or $p_{j+1}$, respectively. In other words, we assume $X_{j}$, the data at available resolution $S_j$, could be obtained as coarsening of the true-model data $\mathcal{X}$, and that each of the intermediate resolutions can be obtained as coarsening of higher resolutions. Given these assumptions, we have
\[ \frac{1}{n} \begin{bmatrix} y'y & y'X_j \\ X_j'y & X_j'X_j \end{bmatrix} \xrightarrow[a.s.]{} \begin{bmatrix} \omega_{yy} & \mathbf{\omega}_{y \mathbf{x}_j} \\ \mathbf{\omega}_{\mathbf{x}_j y} & \Omega_j \end{bmatrix}. \]
Note that we do not assume $\mathcal{L}_j = I_p$ necessarily for some $j$. In other words, if $S$ is the resolution of the data corresponding to the true regression coefficient $b$, it may hold that $S_j \neq S$ for all $j$. The goal is to find the ``best approximation" of $b$ at the different predetermined resolutions.

We consider without loss of generality the overall model
\[ y = X_1 \theta_1 + X_2 \theta_2 + \varepsilon, \] 
where $X_1 = \mathcal{X} \mathcal{L}_1$ and $X_2 = \mathcal{X} \mathcal{L}_2$, respectively. We will assume that the prior for $\theta_j$ when using component module $j$ has positive density on a neighborhood of the corresponding pseudo-true parameter value $\theta_j^*$. At the first stage, we use model
\[y = X_1 \theta_1 + \varepsilon_1 \]
with $\varepsilon_1 \sim N(0, \sigmasq_1)$. For now, we assume $\sigmasq_1$ is known. In large samples, at this resolution and with our assumptions, the Bayesian posterior will be approximated by $N(\beta^*_1, \sigmasq_1 \Omega_1^{-1})$, where $\beta^*_1 = \Omega_1^{-1} \mathbf{\omega}_{\mathbf{x}_1 y} = (\mathcal{L}_1' \Omega \mathcal{L}_1)^{-1} \mathcal{L}_1' \Omega b $ is the pseudo-true value of the regression coefficients $b$ at low resolution $S_1$, and $\Omega_1 = \mathcal{L}_1' \Omega \mathcal{L}_1$ is the second derivative of the log-likelihood of this module (this corresponds to results in \citealt{bda3}, \citealt{geweke2005}, \citealt{kleijn2012}). 
At the second stage, we fix $\theta_1$ and consider the model
\[y = X_1 \theta_1 + X_2 \theta_2 + \varepsilon_2 \]
with $\varepsilon_2 \sim N(0, \sigmasq_2)$. Here, we assume both $\theta_1$ and $\sigmasq_2$ are known. In large samples, at this resolution and with our assumptions, the Bayesian posterior of $\theta_2$ will be approximated by $N(\theta^*_2, \sigmasq_2 \Omega_2^{-1})$, where $\theta^*_1 = \Omega_2^{-1} \mathbf{\omega}_{\mathbf{x}_2 y} - \Omega_2^{-1} L_1 \Omega_1 \theta_1 = (\mathcal{L}_2' \Omega \mathcal{L}_2)^{-1} \mathcal{L}_2' \Omega b - (\mathcal{L}_2' \Omega \mathcal{L}_2)^{-1} \mathcal{L}_2' \Omega \mathcal{L}_1 \theta_1 = \beta^*_2 - (\mathcal{L}_2' \Omega \mathcal{L}_2)^{-1} \mathcal{L}_2' \Omega \mathcal{L}_1 \theta_1 = \beta^*_2 - L_1 \theta_1$ is the difference between $\beta^*_2$, i.e. the pseudo-true value of the regression coefficients $b$ at resolution $S_2$, and the rescaled $\theta_1$. $\Omega_2 = \mathcal{L}_2' \Omega \mathcal{L}_2$ is the second derivative of the log-likelihood of this module. 

The modular posterior is by definition (\ref{bmms_general:posterior}) the product of the two component modules' posteriors. Since both are normal in large samples, the modular posterior will also be approximately normal in large samples, with mean $\bar{\mu}_{1:2}$ and covariance matrix $ \bar{\Sigma}_{1:2}$:
\begin{equation*} 
\begin{aligned}
\begin{array}{ccc}
\bar{\mu}_{1:2} = 
\begin{bmatrix}
\theta_1^* \\
\theta_2^*
\end{bmatrix} = \begin{bmatrix}
\beta_1^* \\
\beta_2^* - L_1 \beta_1^*
\end{bmatrix} & \quad & 
\bar{\Sigma}_{1:2} =
\begin{bmatrix}
\sigmasq_1\Omega_1^{-1} & - \sigmasq_1 \Omega_1^{-1}Q'\\
- \sigmasq_1 L_1 \Omega_1^{-1} & \sigmasq_2 \Omega_2^{-1} +  \sigmasq_1 L_1 \Omega_1^{-1} L_1'
\end{bmatrix}
\end{array}
\end{aligned}
\end{equation*}
since $L_1 = (\mathcal{L}_2' \Omega \mathcal{L}_2)^{-1} \mathcal{L}_2' \Omega \mathcal{L}_1$. This follows from the properties of the multivariate normal distribution. Results analogous to the standard case of linear regression models follow when $\sigmasq_1$ and $\sigmasq_2$ are unknown. Similarly, asymptotic normality of the modular posterior is preserved when $K>2$.

\section{Appendix to Section 5}\label{appendix:key_models}
\subsection{Implemented models}
\begin{itemize}
	\item \bmms (in simulated data analysis): Bayesian Modular \& Multiscale Regression using Algorithm \ref{comp:algorithm_fix} with changepoint-detection modules
	\item \bmms (1): uses Algorithm \ref{comp:algorithm_fix} with \textit{Spike \& Slab} modules
	\item \bmms (2): uses Algorithm \ref{comp:algorithm_fix} with modules on 2D representation of Appendix \ref{appendix:soi_reg}
	\item \textit{Linear Model, L1} or \textit{Lasso}: lasso regression \citep{lasso} using cross-validation for $\lambda$ (from R package \texttt{glmnet})
	\item \textit{Linear Model, L2} or \textit{Ridge}: ridge regression \citep{ridge} using cross-validation for the ridge parameter (from R package \texttt{MXM})
	\item \textit{Modular MS, L1}: Modular (frequentist) Multiscale where each module is a L1-penalized regression. $\lambda$ found via 10-fold cross-validation
	\item \textit{Wavelets: Daub=2, L1 (frequentist)}: regression using wavelet-transformed data (using Daubechies extremal phase wavelets of order 2). See \cite{nason, lasso_wavelet} 
	\item \textit{Fourier regression}: functional linear model using Fourier-basis representation of the data, using cross-validation to estimate the number of basis elements
	\item \textit{Bayesian Changepoint}: regression coefficient $\beta$ assumed as in Equation \ref{data:changepoint}, with unknown $H$ and using Reversible-Jumps MCMC of \cite{green1995} to sample from the posterior
	\item \textit{Bayesian VS} or \textit{Spike \& Slab}: Bayesian Variable Selection using g-priors \citep{marin_robert}
	\item \textit{Functional PCR}: Functional Principal Component Regression, with cross-validated number of bases (from R package \texttt{refund}).
\end{itemize}

\subsection{BM\&Ms for scalar-on-image regression} \label{appendix:soi_reg}
In the general context of scalar-on-tensor regression, each observation $i$ corresponds to a scalar response $y_i$, and the predictor is a $D$-dimensional array $\mathbf{X}_i \in \mathbb{R}^{p_1 \times \cdots \times p_D}$. In the HCP data we consider, the response is binary and the tensor corresponds to an image, i.e. $D=2$. Each element $x_{i, (j_1, \dots, j_D)}$ of $\mathbf{X}_i$ is associated to $y_i$ via the corresponding regression coefficient $\beta_{j_1, \dots, j_D}$. These coefficients can be collected in a tensor $\mathbf{B}$ having the same size as $\mathbf{X}_i$. Thus in scalar-on-image regression we are estimating a coefficient \textit{matrix} (or image). The resulting regression model can be written, for $i=1, \dots, n$, as:
\[ y_i = vec(\mathbf{X_i})vec(\mathbf{B}) + \varepsilon_i,\]
which means we could implement standard methods for linear regression for the estimation of $vec(\mathbf{B})$. However, these are likely to be poor performers given the data dimensionality. Methods for scalar-on-tensor regression typically assume that there exists some simplifying decomposition for the $\mathbf{B}$ tensor \citep{btr, zhou_li_zhu, Li2018}, which reduces the number of free parameters. Instead, we consider a modular approach where each model corresponds to assuming $\mathbf{B}$ is a ``step-surface." If $D=1$, this reduces to the modules of Section \ref{data:simulated}, where the goal is to estimate the number of splits of the step functions, and identify their locations. With $D\geq 2$, we use a hierarchical Voronoi tesselation to represent the surfaces, with the goal of estimating their number and the locations of their centers. Figure \ref{voronoi} shows how a square image can be decomposed into increasingly finer-scale regions by using a hierarchical Voronoi structure. We thus represent a single-scale scalar-on-image regression model as
\[ y_i = \mathbf{X_i} \mathcal{L}_j \beta_j + \varepsilon_i, \]
where, for $i = 1, \dots, n$ we apply a coarsening operator $\mathcal{L}_j$ to the raw image $\mathbf{X_i}$, which will result in the estimation of a low resolution $\beta_j$ and the corresponding step-surface image $\mathcal{L}_j \beta_j$. Using this as a component module of \bmms, we can implement the overall model
\[ y_i = \mathbf{X_i} ( \mathcal{L}_1 \theta_1 + \dots + \mathcal{L}_K \theta_K) + \varepsilon_i, \]
following Section \ref{modular_linear}. Given each component module is at low resolution, i.e. the dimension of $\theta_j$ is small, we can use default Normal priors.

\begin{figure}
	\centering
	\includegraphics[width=0.2\textwidth]{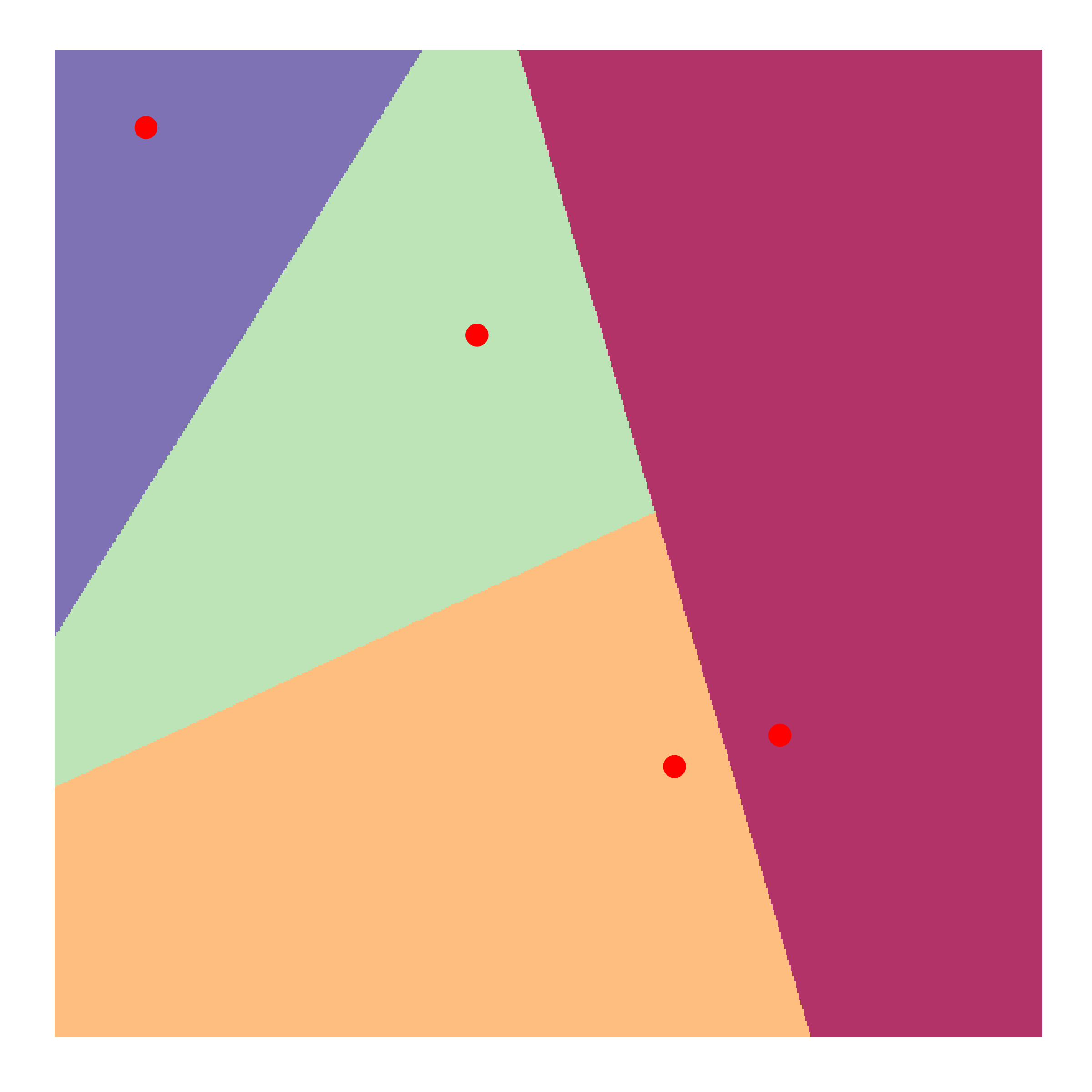}
	\includegraphics[width=0.2\textwidth]{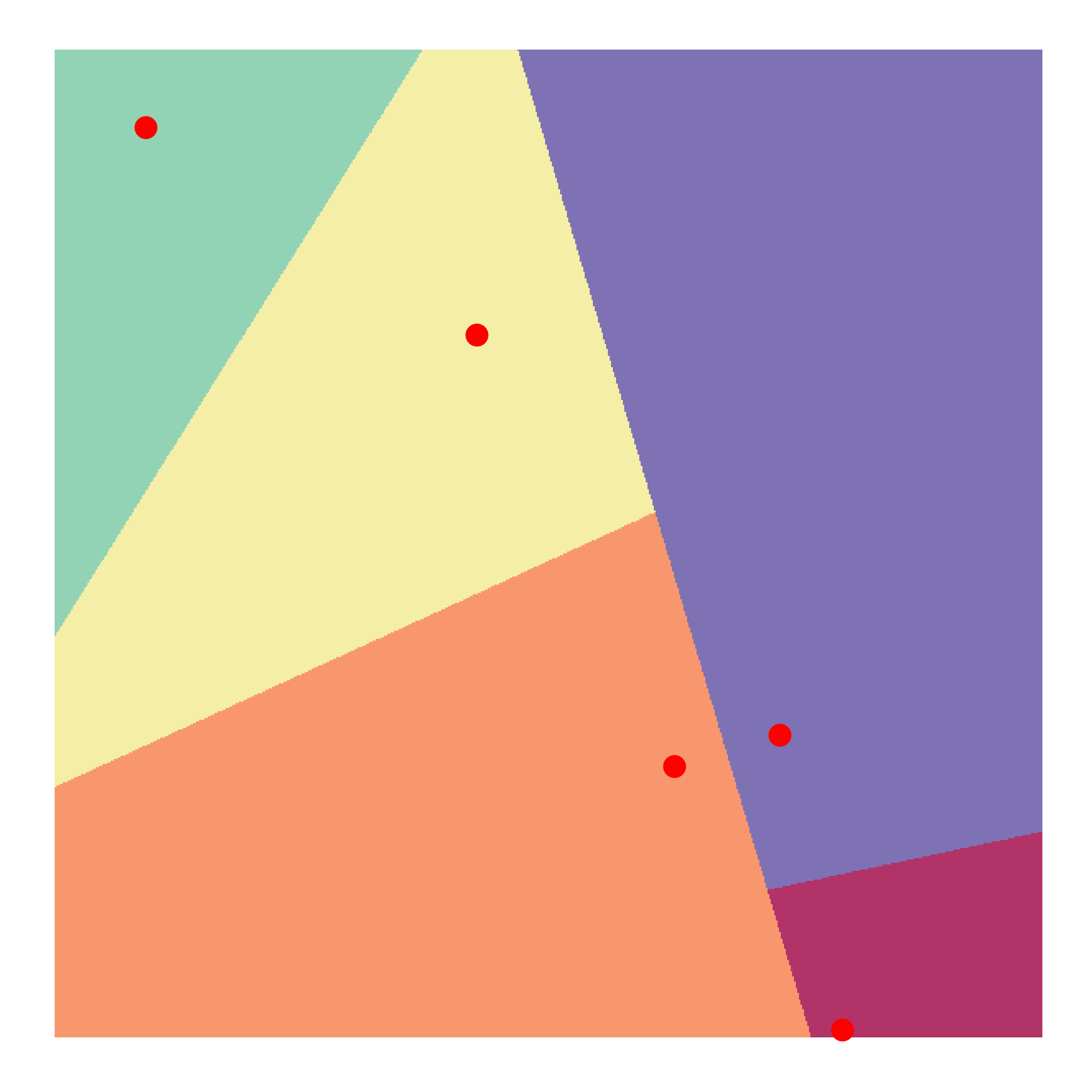}
	\includegraphics[width=0.2\textwidth]{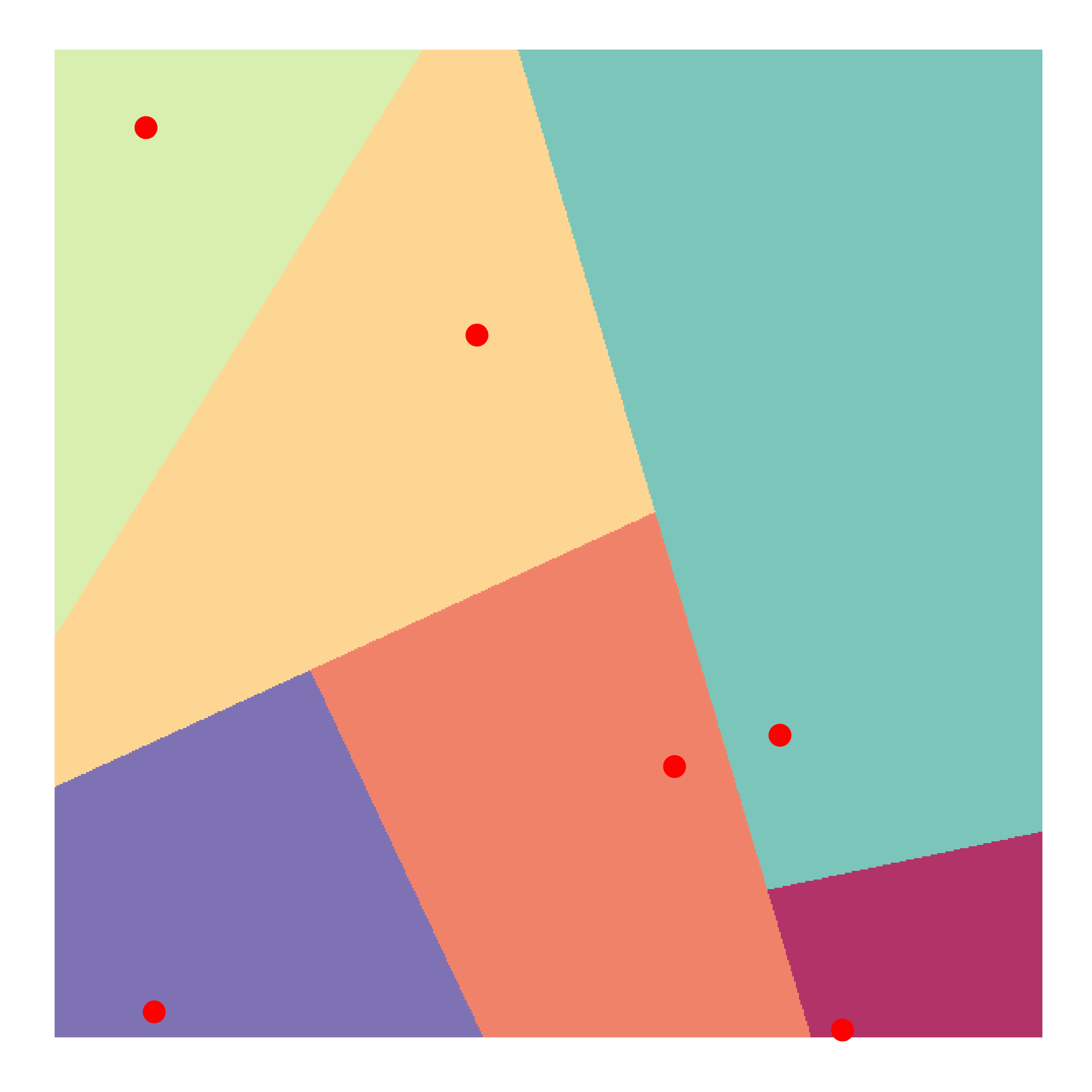}
	\includegraphics[width=0.2\textwidth]{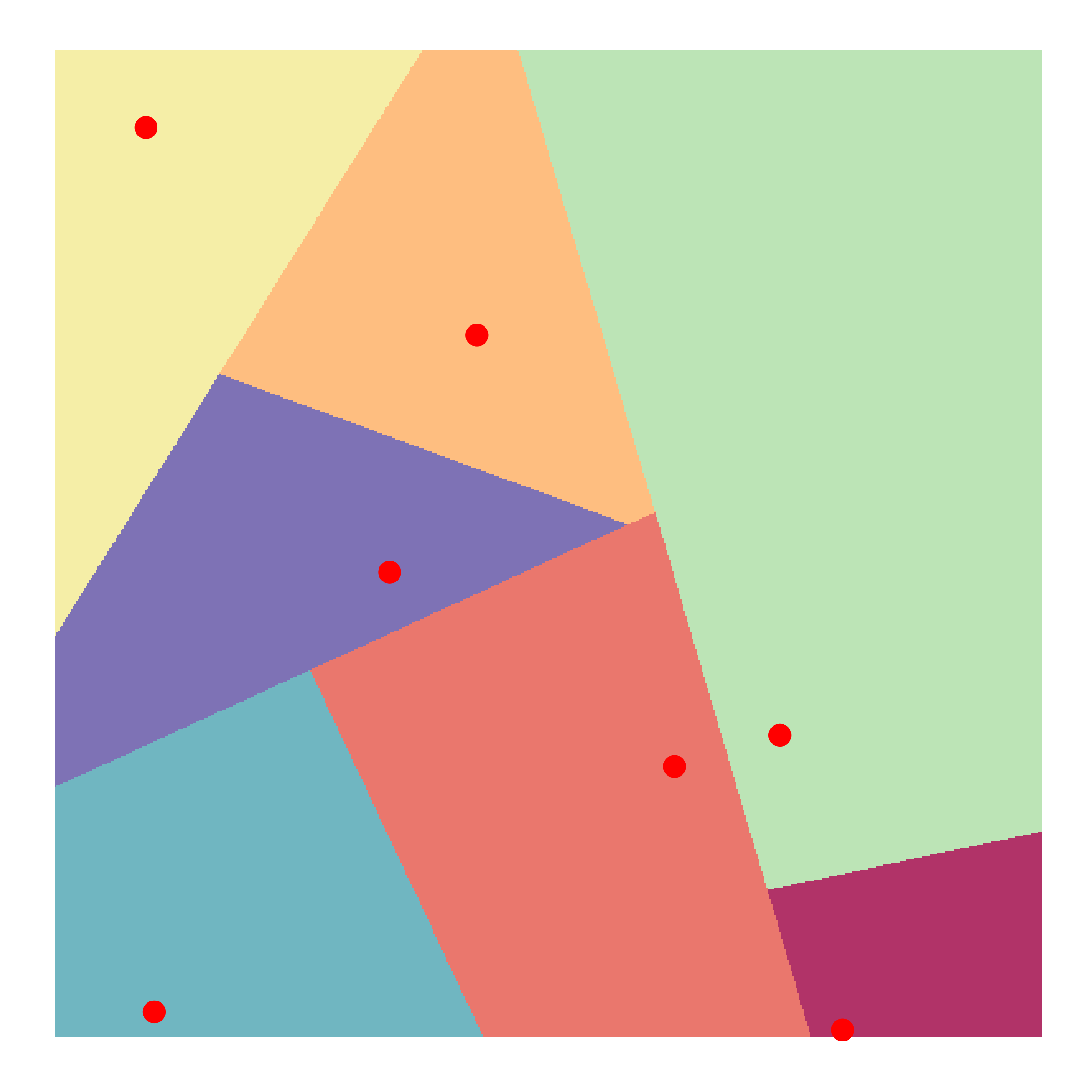}
	\caption{Hierarchical Voronoi representation of ``step-surfaces:" in each of the four images, each color corresponds to a single value.}\label{voronoi} 
\end{figure}

%% file: 92_bmms_supplement.tex
\section{A toy example} \label{appendix:toy}
Suppose only two measurements are taken from a sensor at times $t_1$ and $t_2$, to be used as inputs in regression. In our notation, $S_p = \{ t_1, t_2 \}$,  $|S_p| = p = 2$, $X_{S_p} = \begin{bmatrix} x_{t_1} & x_{t_2} \end{bmatrix} $. For simplicity, we call $x_1 = x_{t_1}$, $x_2 = x_{t_2}$, $X_1 = X_{S_p}$, $X_0 = x_{t_1} + x_{t_2} $, and we assume its correlation structure depends on parameter $r$ as follows:
\[ 
C(r) = \frac{1}{n}X_1'X_1 = \begin{bmatrix} 1 & r \\ r & 1 \end{bmatrix}
\]
which implies $ \frac{1}{n} X_0'X_0 = 2 + 2r $. We fix $\beta_1, \beta_2, \varepsilon \sim N(0, \sigmasq I_n)$ and set $\bar{\beta}_1 = \frac{1}{2}(\beta_1 + \beta_2)$. We then consider two models
\begin{equation*}
\begin{array}{ccc}
\text{(i)}\quad y = x_{1} \beta_1 + x_{2}\beta_2 + \varepsilon & \qquad & \text{(ii)}\quad y = (x_{1}+ x_{2})\bar{\beta}_1 + \varepsilon
\end{array}
\end{equation*}
Models (i) and (ii) consider the data at the high and low resolutions, respectively. Note that the KL divergence of (ii) from (i) is increasing in $| \beta_1 - \beta_2|$ and decreasing on $r$: high observed correlations (large positive $r$) between covariates make the lower resolution model a good approximation of the high resolution one, and thus we might prefer it, given its increased parsimony.

The KL divergence of the low-res likelihood from the high resolution one is
\[ KL(r) = \frac{n}{2 \sigmasq} (\bar{\beta} - \beta)'C(r)(\bar{\beta} - \beta) \]
which is a decreasing function of $r$, since $ \frac{\delta }{\delta r}KL(r) = \frac{n}{\sigmasq} (\bar{\beta}_1 - \beta_1)(\bar{\beta}_1 - \beta_2) \leq 0$. 


We implement \bmms as in Section \ref{modular_linear}, and get:\footnote{Full derivations in Section \ref{appendix:toy:mse}.}
\begin{equation*}\begin{array}{ccc}
\Sigma_0 = (2(r+1)(n+1))^{-1} & \qquad & 
\Sigma_1 = \frac{1}{n+1}\begin{bmatrix}
1 & r \\ r & 1
\end{bmatrix}^{-1} \\
\mu_0 = \frac{x_1'y  + x_2'y}{2(r+1)(n+1)} \xrightarrow[n \to \infty]{}  \frac{\beta_1 + \beta_2}{2} & \qquad & 
\mu_1 = \frac{1}{n+1}\begin{bmatrix}
1 & r \\ r & 1
\end{bmatrix}^{-1} \begin{bmatrix}
x_1'e_1 \\ x_2'e_1
\end{bmatrix} 
\xrightarrow[n \to \infty]{}  \begin{bmatrix}
\frac{\beta_1 - \beta_2}{2} \\
\frac{\beta_2 - \beta_1}{2}
\end{bmatrix}
\end{array}
\end{equation*}
Note how $\mu_0$ roughly corresponds to the average of the high-resolution coefficient vector, whereas $\mu_1$ -- which is interpreted as the added detail from the higher resolution -- to half differences. Finally, the asymptotic variance of $\theta_1 = (\theta_{11}, \theta_{12})$ becomes
\begin{align*}
AVar(\theta_1)=\frac{\sigmasq}{2(1-r)(r+1)} \begin{bmatrix} 3-r & 3r-1 \\  3r-1 & 3-r \end{bmatrix}
\end{align*}
and this shows that\footnote{Since the determinant of the 2x2 matrix is $8(1-r)(1+r)$ hence the determinant of the asymptotic variance is $\frac{8}{(1-r)(1+r)}$} $r \approx 1$ makes the higher resolution worthless. Otherwise, the coefficient vector at the highest resolution can be estimated consistently with \bmms via $\mu_{\beta_M}$:
\[\mu_{\beta_M} 
= L_0\mu_0  + \mu_1 = \begin{bmatrix}
1  & 1 & 0 \\ 1 & 0 & 1
\end{bmatrix} \begin{bmatrix} \mu_0 \\ \mu_1 \end{bmatrix} =  \begin{bmatrix} \mu_0  + \mu_{1,1} \\ \mu_0 + \mu_{1,2} \end{bmatrix} \xrightarrow[n \to \infty]{} \begin{bmatrix}
\frac{\beta_1 + \beta_2}{2} \\ \frac{\beta_1 + \beta_2}{2}
\end{bmatrix} + \begin{bmatrix}
\frac{\beta_1 - \beta_2}{2} \\
\frac{\beta_2 - \beta_1}{2}
\end{bmatrix} = \begin{bmatrix} \beta_1 \\ \beta_2 \end{bmatrix}  \]
We now consider the finite-sample frequentist MSE of
\begin{align*}
\mu^c_{\beta_{1,M}} &= L_0 \mu_0 + c \mu_1 = L_0 \mu_0 + cl\hat{\beta} - c\Sigma_1 X_1' X_0 \mu_0
= (1-cl)L_0 \mu_0 + c l\hat{\beta}
\end{align*}
with $c \in \{0,1\}$ and $l = \frac{n}{n+1}$. If we select $c=0$, we are estimating $\beta$ through $L_0\mu_0$, meaning that we completely discard the contribution of the high resolution. The MSE of this estimator for $c=0$ and $c=1$ is, respectively:
\begin{align*}
\text{MSE}(\mu^{c=0}_{\beta_{1,M}}) &= ((l-2)^2 + l^2) \frac{\beta_1^2 + \beta_2^2}{4} + l(l-2) \beta_1 \beta_2 + \frac{\sigmasq_1}{n(1+r)} \\
\text{MSE}(\mu^{c=1}_{\beta_{1,M}}) &= (1-l)^2 ((l-2)^2 + l^2)\frac{\beta_1^2 + \beta_2^2}{4} + (1-l)^2 l (l-2) \beta_1 \beta_2 + \frac{2 l^2 \sigmasq_1}{n(1+r)(1-r)}
\end{align*}
First, for $n \to \infty$ the bias term approaches $0$ only for $c=1$, whereas the variance will decrease in both cases with $n$. A more relevant scenario is $r \approx 1$ and/or $\beta_1 \approx \beta_2$: if $c=1$ the variance diverges when $r \to 1$. In other words, if $r$ is large, considering the two measurements separately leads to a large expected error. Similarly, $\beta_1 \approx \beta_2$ results in 
\begin{align*}
\text{MSE}(\mu^{c=0}_{\beta_{1,M}})_{\beta_1\approx \beta_2} &= 2(1-l)^2{\beta_1^2} + \frac{\sigmasq_1}{n(1+r)} \\
\text{MSE}(\mu^{c=1}_{\beta_{1,M}})_{\beta_1\approx  \beta_2} &= 2(1-l)^4{\beta_1^2} + \frac{2 l^2 \sigmasq_1}{n(1+r)(1-r)}
\end{align*}
meaning that the bias term is almost equalized, but favoring $c=0$, whereas the comparison on variance entirely depends on $r$: the closer the two coefficients $\beta_1$ and $\beta_2$ are to each other, the closer to zero $r$ must be to make it worth it to consider the high resolution. Ultimately, this shows how considering the data at the highest resolution might be counterproductive. An alternative way to visualize why $c=1$ may be suboptimal is to look at $\mu^c_{\beta_{1,M}}$ (with $c \in [0,1]$ ) as an estimator that shrinks towards the lower resolution coefficient function.

\subsection{Modules for BM\&Ms in example} \label{appendix:toy:modules}
\textbf{Module 0:}
\[ y = \theta_0 (x_{1} + x_{2}) + \varepsilon_0 \textalt{where} \varepsilon_0 \sim N(0, \sigmasq_0 I_n ) \]
with prior parameters $m_0 = 0$ and $M_0 = n ((x_1 + x_2)'(x_1 + x_2))^{-1} = \frac{1}{2} (r+1)^{-1}$. We get
\begin{align*}
\Sigma_0 &= (2(r+1) + 2n(r+1))^{-1} = (2(r+1)(n+1))^{-1}\\
\mu_0 &= \frac{(x_1 + x_2)'y}{2(r+1)(n+1)} = \frac{n \frac{x_1'y}{n} + n\frac{x_2'y}{n}}{2(r+1)(n+1)} \xrightarrow[n \to \infty]{}  \frac{\beta_1 + \beta_2}{2}
\end{align*}
\textbf{Module 1:}
\[ e_1 = \begin{bmatrix} x_1 & x_2 \end{bmatrix}\begin{bmatrix} \theta_{11} \\ \theta_{12} \end{bmatrix} + \varepsilon_1 \textalt{where} \varepsilon_1 \sim N(0, \sigmasq_1) \]
where $e_1 = y - (x_1 + x_2)\frac{(x_1 + x_2)'y}{2(r+1)(n+1)}$. In this case we set the prior parameters as
\begin{align*}
\begin{array}{ccc} m_1 =\begin{bmatrix}
m_{11} & m_{12}
\end{bmatrix} = \begin{bmatrix}
0 & 0
\end{bmatrix} & \qquad & 
M_1 = n\begin{bmatrix}
1 & r \\ r & 1 
\end{bmatrix}^{-1},
\end{array} 
\end{align*}
hence the posterior parameters are
\begin{align*}
\Sigma_1 &= \frac{1}{n+1}\begin{bmatrix}
1 & r \\ r & 1
\end{bmatrix}^{-1}\\
\mu_1 &= \frac{1}{n+1}\begin{bmatrix}
1 & r \\ r & 1
\end{bmatrix}^{-1} \begin{bmatrix}
x_1'e_1 \\ x_2'e_1
\end{bmatrix} \\ &= \frac{1}{n+1}\begin{bmatrix} 1 & r \\ r & 1 \end{bmatrix}^{-1} \begin{bmatrix} x_1'y - x_1'(x_1 + x_2)\frac{(x_1 + x_2)'y}{2(r+1)(n+1)} \\ x_2'y - x_2'(x_1 + x_2)\frac{(x_1 + x_2)'y}{2(r+1)(n+1)} \end{bmatrix}\\
&\xrightarrow[n \to \infty]{} \begin{bmatrix} 1 & r \\ r & 1 \end{bmatrix}^{-1} \begin{bmatrix} \beta_1 + r\beta_2 - (r+1)\frac{\beta_1 + \beta_2}{2} \\ r\beta_1 + \beta_2 - (r+1)\frac{\beta_1 + \beta_2}{2} \end{bmatrix}\\ &= 
\frac{1}{r+1}\begin{bmatrix} 1 & -r \\ -r & 1 \\ \end{bmatrix} \begin{bmatrix} \frac{1}{2}\beta_1  - \frac{1}{2}\beta_2 \\ \frac{1}{2}\beta_2  - \frac{1}{2}\beta_1 \end{bmatrix} 
\frac{1}{r+1}\begin{bmatrix} 1 & -r \\ -r & 1 \\
\end{bmatrix} (\beta - L_0C_0 \beta) = \begin{bmatrix}
\frac{\beta_1 - \beta_2}{2} \\
\frac{\beta_2 - \beta_1}{2}
\end{bmatrix},
\end{align*}
where $C_0 = \begin{bmatrix}\frac{1}{2} & \frac{1}{2}\end{bmatrix}$.  $\Sigma_1$ is the posterior covariance of $\theta_1|\theta_0$, whereas the asymptotic variance of $\theta_1 = (\theta_{11}, \theta_{12})$ is
\begin{align*}
AVar(\theta_1) &= \Omega_1^{-1} + \sigmasq L_0 \Omega_0^{-1} L_0' \\
&= \sigmasq \left( \begin{bmatrix} 1& r \\ r & 1 \end{bmatrix}^{-1} + \frac{1}{2(r+1)}L_0 L_0' \right)\\
 &= \sigmasq \left( \frac{1}{(1-r)(r+1)}\begin{bmatrix} 1& -r \\ -r & 1 \end{bmatrix} + \frac{1}{2(r+1)}\begin{bmatrix} 1& 1 \\ 1 & 1 \end{bmatrix}\right)\\
&= \frac{\sigmasq}{r+1} \left( \frac{1}{1-r}\begin{bmatrix} 1& -r \\ -r & 1 \end{bmatrix} + \frac{1}{2}\begin{bmatrix} 1& 1 \\ 1 & 1 \end{bmatrix}\right) = \frac{\sigmasq}{r+1} \left( \begin{bmatrix} \frac{1}{1-r}& \frac{-r}{1-r} \\ \frac{-r}{1-r} & \frac{1}{1-r} \end{bmatrix} + \begin{bmatrix} \frac{1}{2} & \frac{1}{2} \\ \frac{1}{2} & \frac{1}{2} \end{bmatrix}\right) \\
&= \frac{\sigmasq}{r+1} \begin{bmatrix} \frac{2 + 1-r}{2(1-r)} & \frac{-2r + 1-r}{2(1-r)} \\  \frac{-2r + 1-r}{2(1-r)} & \frac{2 + 1-r}{2(1-r)}  \end{bmatrix} = \frac{\sigmasq}{r+1} \begin{bmatrix} \frac{3-r}{2(1-r)} & \frac{1-3r}{2(1-r)} \\  \frac{1-3r}{2(1-r)} & \frac{3-r}{2(1-r)}  \end{bmatrix} \\
&=\frac{\sigmasq}{2(1-r)(r+1)} \begin{bmatrix} 3-r & 3r-1 \\  3r-1 & 3-r \end{bmatrix}
\end{align*}
\subsection{MSE for $\mu^c_{\beta_1, M}$} \label{appendix:toy:mse}
The expected value of the modular estimator for $\beta$ is
\begin{align*}
E\left[ \mu^c_{\beta_{1,M}} \right] &= E\left[ (1-cl)L_0 \mu_0 + c l\hat{\beta} \right] \\
&= E\left[ l(1-cl) L_0 (X_0'X_0)^{-1}X_0'y + c l\hat{\beta} \right] \\
&= E\left[ l(1-cl) L_0 (X_0'X_0)^{-1}X_0'(X_1 \beta + \varepsilon_1) \right] + c l\beta
\end{align*}
\begin{align*}
&= l(1-cl) L_0 (X_0'X_0)^{-1}X_0'X_1 \beta + c l\beta \\
&= (l(1-cl) L_0 (X_0'X_0)^{-1}X_0'X_1  + c l)\beta \\
&= (l(1-cl) \begin{bmatrix}
\frac{1}{2n(r+1)} \\ \frac{1}{2n(r+1)}
\end{bmatrix} \begin{bmatrix} n(1+r) & n(1+r) \end{bmatrix}  + c l)\beta \\
&= l(1-cl) \begin{bmatrix} \frac{\beta_1 +\beta_2}{2} \\ \frac{\beta_1 + \beta_2}{2} \end{bmatrix} + cl \begin{bmatrix} \beta_1 \\ \beta_2 \end{bmatrix}
\end{align*}
so that the estimator has bias
\begin{align*}
E\left[ \mu^c_{\beta_{1,M}} \right] - \beta &= \frac{1-cl}{2} \begin{bmatrix} l-2 & l \\ l & l-2 \end{bmatrix} \beta 
\end{align*}
which, as expected, is smaller for $l \approx 1$. We also obtain
\begin{align*}
\text{Bias}^2 &= \frac{(1-cl)^2}{4} \beta' \begin{bmatrix} (l-2)^2 + l^2 & 2(l-2)l \\ 2(l-2)l & (l-2)^2 + l^2 \end{bmatrix} \beta \\
&= (1-cl)^2 ((l-2)^2 + l^2)\frac{\beta_1^2 + \beta_2^2}{4} + (1-cl)^2 l (l-2) \beta_1 \beta_2
\end{align*}
We then move to calculating the variance of $\mu_{\beta}$.
\begin{align*}
Var(\mu^c_{\beta_{1,M}}) &= Var( (1-cl) L_0 (X_0'X_0)^{-1}X_0'y + cl (X_1'X_1)^{-1} X_1'y ) \\
&= Var( ( (1-cl) L_0 (X_0'X_0)^{-1} X_0' + cl (X_1'X_1)^{-1}X_1') \varepsilon ) \\
&= \sigmasq_1 ((1-cl) L_0 (X_0'X_0)^{-1}X_0' + cl(X_1'X_1)^{-1} X_1')((1-cl) L_0 (X_0'X_0)^{-1}X_0' + cl(X_1'X_1)^{-1} X_1')' \\
&= \sigmasq_1 ( (1-cl)^2 L_0 (X_0'X_0)^{-1} L_0' + cl(1- cl) L_0 (X_0'X_0)^{-1}X_0'X_1(X_1'X_1)^{-1} \\
& \qquad + cl(1-cl)(X_1'X_1)^{-1}X_1'X_0 (X_0'X_0)^{-1}L_0' + c^2 l^2 (X_1'X_1)^{-1}) 
\end{align*}
\begin{align*}
&= \sigmasq_1 \left( \frac{(1-cl)^2}{2n(r+1)}\begin{bmatrix} 1 & 1 \\ 1& 1 \end{bmatrix} +\frac{cl(1-cl)}{2n(r+1)(1-r^2)} \begin{bmatrix} 1 \\1  \end{bmatrix} \begin{bmatrix} 1+r & 1+r \end{bmatrix} \begin{bmatrix} 1 & -r \\ -r & 1 \end{bmatrix}  +\right. \\
&\qquad + \left. \frac{cl(1-cl)}{n(1-r^2)(r+1)} \begin{bmatrix} 1 & -r \\ -r & 1  \end{bmatrix}\begin{bmatrix} 1+r \\ 1+r \end{bmatrix}\begin{bmatrix} 1 & 1 \end{bmatrix} + \frac{c^2 l^2 }{n(1-r^2)} \begin{bmatrix} 1 & -r \\ -r & 1\end{bmatrix}\right)\\
&= \sigmasq_1 \left( \frac{(1-cl)^2}{2n(r+1)} \begin{bmatrix} 1 & 1 \\ 1& 1 \end{bmatrix} +  \frac{cl(1-cl)}{2n(1-r)(1+r)} \begin{bmatrix} 1-r & 1-r \\ 1-r & 1-r \end{bmatrix} + \right.\\
&\left. \qquad + \frac{cl(1-cl)}{2n(1+r)} \begin{bmatrix} 1 & 1 \\ 1& 1 \end{bmatrix} + \frac{c^2 l^2 }{n(1-r^2)} \begin{bmatrix} 1 & -r \\ -r & 1 \end{bmatrix}  \right) \\
&= \frac{\sigmasq_1}{n(1+r)} \left( \frac{(1-cl)^2}{2} \begin{bmatrix} 1 & 1 \\ 1& 1 \end{bmatrix} +  cl(1-cl)\begin{bmatrix} 1 & 1 \\ 1 & 1 \end{bmatrix} + \frac{c^2 l^2 }{1-r} \begin{bmatrix} 1 & -r \\ -r & 1 \end{bmatrix}  \right)
\end{align*}
And then finally
\begin{align*}
\text{Tr}(\text{Var}(\mu^c_{\beta_{1,M}})) &= \frac{\sigmasq_1}{n(1+r)} \left( (1-cl)(1+cl) + \frac{2c^2 l^2}{1-r}  \right)
\end{align*}
So that the modular estimator has mean square error:
\begin{align*}
\text{MSE}(\mu^c_{\beta_{1,M}}) & = (1-cl)^2 ((l-2)^2 + l^2)\frac{\beta_1^2 + \beta_2^2}{4} + (1-cl)^2 l (l-2) \beta_1 \beta_2 +\\
& \qquad + \frac{\sigmasq_1}{n(1+r)} \left( (1-cl)(1+cl) + \frac{2c^2 l^2}{1-r}  \right)
\end{align*}